\documentclass[english,11pt]{article}
\pdfoutput=1
\usepackage[letterpaper]{geometry}
\usepackage{hyperref}

\usepackage{amsmath,nicefrac}
\usepackage{epsfig}
\usepackage{amsthm}
\usepackage{amssymb,mathrsfs}
\usepackage{xspace}
\usepackage{soul}
\usepackage{latexsym}
\usepackage[dvipsnames]{xcolor}

\usepackage{subcaption}

\usepackage[capitalise, noabbrev]{cleveref}
\crefname{subsection}{Subsection}{Subsections}
\crefname{section}{Section}{Sections}
\crefname{figure}{Figure}{Figures}
\crefname{equation}{Equation}{Equations}
\crefname{algorithm}{Algorithm}{Algorithms}
\crefname{algocf}{Algorithm}{Algorithms}

\PassOptionsToPackage{vlined, boxruled}{algorithm2e}
\usepackage{algorithm2e}

\makeatletter
 \theoremstyle{plain}
 \newtheorem{thm}{\protect\theoremname}[section]
 \theoremstyle{plain}
 
 \theoremstyle{plain}

 \newtheorem{property}{Property}
 
 \theoremstyle{remark}
 
 \theoremstyle{plain}
 \newtheorem{prop}[thm]{\protect\propositionname}
 \ifx\proof\undefined
 \newenvironment{proof}[1][\protect\proofname]{\par
 	\normalfont\topsep6\p@\@plus6\p@\relax
 	\trivlist
 	\itemindent\parindent
 	\item[\hskip\labelsep\scshape #1]\ignorespaces
 }{%
 	\endtrivlist\@endpefalse
 }
 \providecommand{\proofname}{Proof}
 \fi
 \theoremstyle{plain}
 \newtheorem{lem}[thm]{\protect\lemmaname}
 
 \theoremstyle{plain}
 
 \theoremstyle{plain}
 \newtheorem{cor}[thm]{\protect\corollaryname}
 \theoremstyle{definition}
 \newtheorem{defn}[thm]{\protect\definitionname}

\makeatother

\usepackage{babel}
\providecommand{\Observationname}{Observation}
\providecommand{\corollaryname}{Corollary}
\providecommand{\definitionname}{Definition}
\providecommand{\lemmaname}{Lemma}
\providecommand{\propositionname}{Proposition}
\providecommand{\remarkname}{Remark}
\providecommand{\theoremname}{Theorem}

	\global\long\def\ALG{\text{\sc ALG}\xspace}
	
	\global\long\def\OPT{\text{\sc OPT}\xspace}
	
	\global\long\def\ST{\text{\sc st}\xspace}

	\global\long\def\SF{\text{\sc sf}\xspace}
	
	\global\long\def\FL{\text{\sc fl}\xspace}
	
	\global\long\def\ON{\text{\sc ON}\xspace}
	
	\global\long\def\OST{\text{\sc ON}^{\ST}\xspace}
	
	\global\long\def\OSF{\text{\sc ON}^{\SF}\xspace}
	
	\global\long\def\OFL{\text{\sc ON}^{\FL}\xspace}
	
	\global\long\def\PC{\text{\sc PC}\xspace}
	
	\global\long\def\PCST{\text{\sc PC}^{\ST}\xspace}
	
	\global\long\def\PCSF{\text{\sc PC}^{\SF}\xspace}
	
	\global\long\def\PCFL{\text{\sc PC}^{\FL}\xspace}
	
	\global\long\def\ADV{\text{\sc ADV}\xspace}

	\global\long\def\S{\mathcal{S}}

\newcommand{\eat}[1]{}

\newcommand{\alg}{\text{\sc alg}\xspace}
\newcommand{\opt}{\text{\sc opt}\xspace}
\newcommand{\mst}{{\sc mst}\xspace}
\newcommand{\err}{\lambda}

\newcommand{\amort}{\alpha}

\newcommand{\npry}{b}

\newcommand{\pry}[1]{\alpha_{#1}}

\newcommand{\reqs}{R}
\newcommand{\mreqs}[1][]{T^{#1}}
\newcommand{\req}{r}
\newcommand{\preds}{\hat{R}}
\newcommand{\mpreds}[1][]{\hat{T}^{#1}}
\newcommand{\pred}{\hat{r}}

\newcommand{\pr}[1]{\left(#1\right)}

\newcommand{\pc}[1]{\left\{#1\right\}}
\newcommand{\pa}[1]{\left|#1\right|}

\newcommand{\ceil}[1]{\left\lceil#1\right\rceil}
\newcommand{\floor}[1]{\left\lfloor#1\right\rfloor}

\SetCommentSty{mycommfont}

\SetFuncSty{myfuncsty}

\SetKwInOut{Input}{Input}
\SetKwProg{Initialization}{Initialization}{}{}
\SetKwProg{Fn}{Function}{}{end}
\SetKwProg{EFn}{Event Function}{}{end}

\SetKwFunction{Partial}{Partial}
\SetKwFunction{UponRequest}{UponRequest}

\LinesNumbered \RestyleAlgo{boxruled}

\date{}

\makeatletter
\newcommand{\algorithmfootnote}[2][\footnotesize]{%
	\let\old@algocf@finish\@algocf@finish
	\def\@algocf@finish{\old@algocf@finish
		\leavevmode\rlap{\begin{minipage}{\linewidth}
				#1#2
		\end{minipage}}%
	}%
}
\makeatother

\begin{document}

\title{Online Graph Algorithms with Predictions}
\date{}

\author{
Yossi Azar\thanks{Supported in part by the Israel Science Foundation (grant No. 2304/20).}\\Tel Aviv University \and
Debmalya Panigrahi\thanks{Supported in part by an NSF grants CCF-1750140 (CAREER Award) and CCF-1955703, and ARO grant W911NF2110230.}\\Duke University \and
Noam Touitou\\Tel Aviv University
}
\maketitle

\begin{abstract}
Online algorithms with predictions is a popular and elegant 
framework for bypassing pessimistic lower bounds in competitive analysis. 
In this model, online algorithms are supplied with future {\em predictions}
and the goal is for the competitive ratio to smoothly interpolate between 
the best offline and online bounds as a function of the prediction error. 

In this paper, we study online graph problems with predictions. 
Our contributions are the following:
\begin{itemize}
    \item The first question is defining prediction error. For graph/metric 
    problems, there can be two types of error, locations that are not 
    predicted, and locations that are predicted but the predicted 
    and actual locations do not coincide exactly. We design a novel 
    definition of prediction error called {\em metric error with outliers}
    to simultaneously capture both types of errors, which thereby generalizes
    previous definitions of error that only capture one of the two error types. 
    \item We give a general framework for obtaining online algorithms with
    predictions that combines, in a ``black box'' fashion, existing online 
    and offline algorithms, under certain technical conditions.
    To the best of our knowledge, this is the first general-purpose tool
    for obtaining online algorithms with predictions.
    \item Using our framework, we obtain tight bounds on the competitive 
    ratio of several classical graph problems as a function of metric error
    with outliers: Steiner tree, Steiner forest, priority Steiner tree/forest, 
    and uncapacitated/capacitated facility location.
\end{itemize}
Both the definition of metric error with outliers and the general framework 
for combining offline and online algorithms are not specific to the problems
that we consider in this paper. We hope that these will be useful for future 
work on other problems in this domain.

\end{abstract}

	
\section{Introduction}
\label{sec:introduction}
Online algorithms has been the paradigm of choice
in algorithms research for handling future uncertainty 
in input data. But, they have often been criticized 
for being overly pessimistic, and their failure to 
distinguish between algorithms with widely different 
empirical performance in many contexts is well documented. 
To overcome this deficiency, 
recent research has proposed augmenting online algorithms
with future {\em predictions} (e.g., from machine learning models),
the goal being to design algorithms that gracefully 
degrade from the best offline bounds for accurate 
predictions to the best online bounds for arbitrary
predictions. Indeed, a robust literature is beginning to 
emerge in this area with applications to 
caching~\cite{lykouris2018competitive,DBLP:conf/soda/Rohatgi20,JiangPS20,DBLP:conf/approx/Wei20},
rent or buy~\cite{purohit2018improving,GollapudiP19,AnandGP20},
scheduling~\cite{purohit2018improving,LattanziLMV20,Mitzenmacher20, DBLP:conf/spaa/Im0QP21}, 
online learning~\cite{DBLP:conf/nips/DekelFHJ17, DBLP:conf/icml/BhaskaraC0P20},
covering problems~\cite{DBLP:conf/nips/BamasMS20},
frequency estimation~\cite{hsu2018learningbased}, 
low rank approximation~\cite{IndykVY19},
metrical task systems~\cite{AntoniadisCEPS20},
auction pricing~\cite{MedinaV17}, 
budgeted allocation~\cite{MirrokniGZ12},
Bloom filters~\cite{mitzenmacher2018model}, etc.
In this paper, we ask: {\em how do we design online graph algorithms
with predictions?} 

For example, consider a cable company that needs to set up a network that can serve clients in a certain area. They can obtain a rough prediction of their future client locations via, e.g., market surveys. This prediction is available upfront, but the actual clients join the service over time, and the network has to be built up as the clients join. 
A natural goal for the cable company is to use the predictions to mitigate uncertainty, and design an algorithm that degrades gracefully with prediction error. 

Another example is a video conferencing platform that has to allocate bandwidth for client meetings. The platform can make predictions on client locations for meetings based on past data (e.g., just their last known location, or by using a more sophisticated machine learning model), but the actual set of users connect to the meeting online. Of course, the prediction could be wrong in terms of the users who actually join the meeting, or even for a correctly predicted user, might get her precise location wrong. Again, the goal would be to smoothly degrade in performance with these types of prediction error.

\subsection{Our Contributions}

{\bf Metric Error with Outliers.}
The first question is how to model prediction error. 
For concreteness, think of the example above, i.e., making predictions on a geographical map; on such a metric space/graph, there are, conceptually, two types of errors.
\begin{itemize}
    \item Some predictions might be entirely wrong; for such cases, the appropriate notion of error is the number of false predictions, or the number of unpredicted requests. 
    This is the error typically implied when solving problems with outliers.
    However, this error notion is unrealistic for metric spaces: even when the prediction for a location is ``correct'', one cannot hope that it would \emph{precisely} match the actual location on the underlying graph/metric space.
    \item Another notion of error, which is able to handle such slightly-perturbed predictions, is the minimum transportation cost of moving from the predicted input to the actual input. (This is the error previously used for metric problems in the prediction model.) The shortcoming of this error notion is that even a single acute misprediction could lead to unbounded transportation cost.
\end{itemize}


It is clear from this discussion that both these types of error are required to faithfully capture prediction (in)accuracy. We introduce a new robust notion of prediction error called {\em metric error with outliers} that subsumes both these types of errors. This notion of error applies to any problem involving predictions on a metric space; with the prediction model becoming very popular in online algorithms, we hope that this definition of error will be useful for other problems in the future.  

For convenience, we define this error on a metric space; this is easily 
extended to a graph using the shortest path metric. Let $M = (V, d)$ 
be a metric space.
Let $\reqs \subseteq V$ be a subset of vertices and $\preds\subseteq V$ 
be a prediction
for $\reqs$. (In general, $\preds$ and $\reqs$ may have different sizes.) 
The {\em metric error with outliers} for $\preds$ is given by $\err = (\Delta, D)$ w.r.t. $\reqs$ if there exist
sets $\mreqs\subseteq \reqs$ and $\mpreds\subseteq \preds$ satisfying $\pa{\mreqs} = \pa{\mpreds}$ 
such that (a) the minimum cost matching between $\mreqs$ and $\mpreds$
on $M$ has cost $D$, and (b) $\pa{\reqs \setminus \mreqs} + \pa{\preds \setminus \mpreds} = \Delta$.
Intuitively, the idea is to define the error as the minimum 
cost matching $D$ between the predicted and actual locations,
but allow {\em outliers} in both the predicted and actual sets that
count toward a numerical error $\Delta$. 
(See \cref{fig:Intro_ViewFigure} for a visualization.) 

We remark that for a given prediction, there are multiple different values of $\Delta$ and $D$ that fit this $(\Delta, D)$-error definition. In particular, for a given $D$, to find the best $\Delta$, simply find the maximum number of pairs of predicted and actual locations whose transportation cost is at most $D$; the total number of remaining predicted and actual locations is $\Delta$. This produces a Pareto frontier of $(\Delta, D)$ pairs (see Fig.~\ref{fig:Intro_ViewFigure}).
Our theorems hold for {\bf all} possible values of $\Delta$ and $D$.
In particular, if we set $D=0$ and define error using only the numerical count of mispredicted locations, or set $\Delta = 0$ and define error using only metric distance between predicted and actual points, our theorems continue to hold for both cases. (These are the two ends of the Pareto frontier.)

\begin{figure}
        \includegraphics[width=0.3\textwidth]{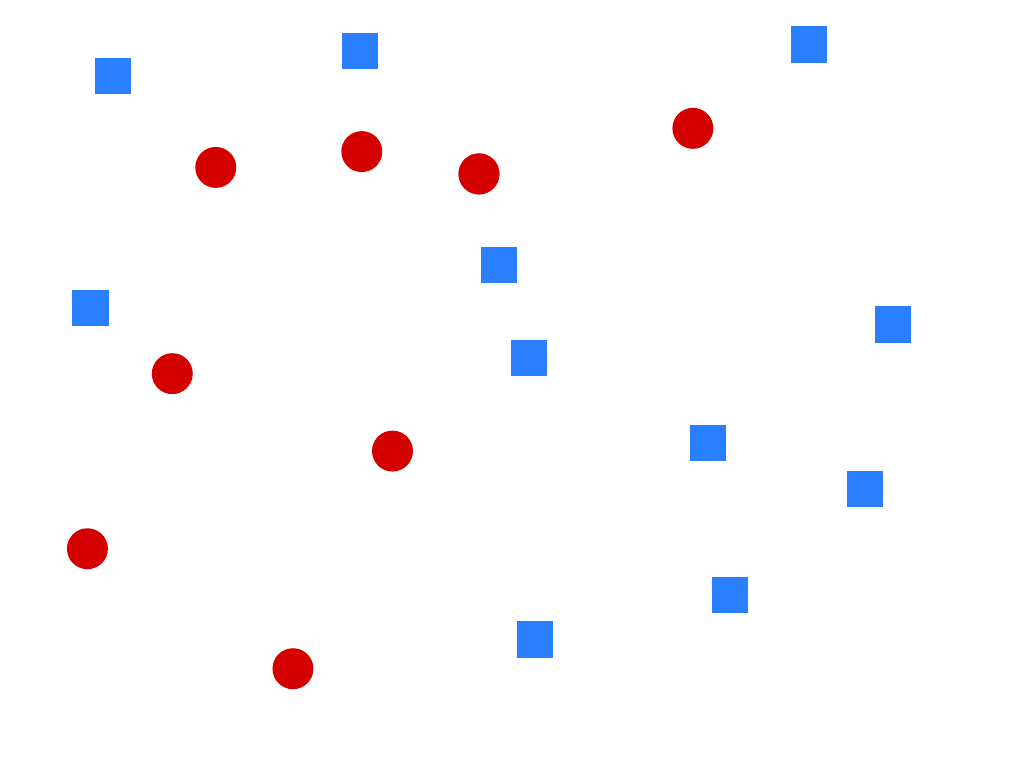}
        \hspace*{\fill}\vline\hspace*{\fill}
        \includegraphics[width=0.3\textwidth]{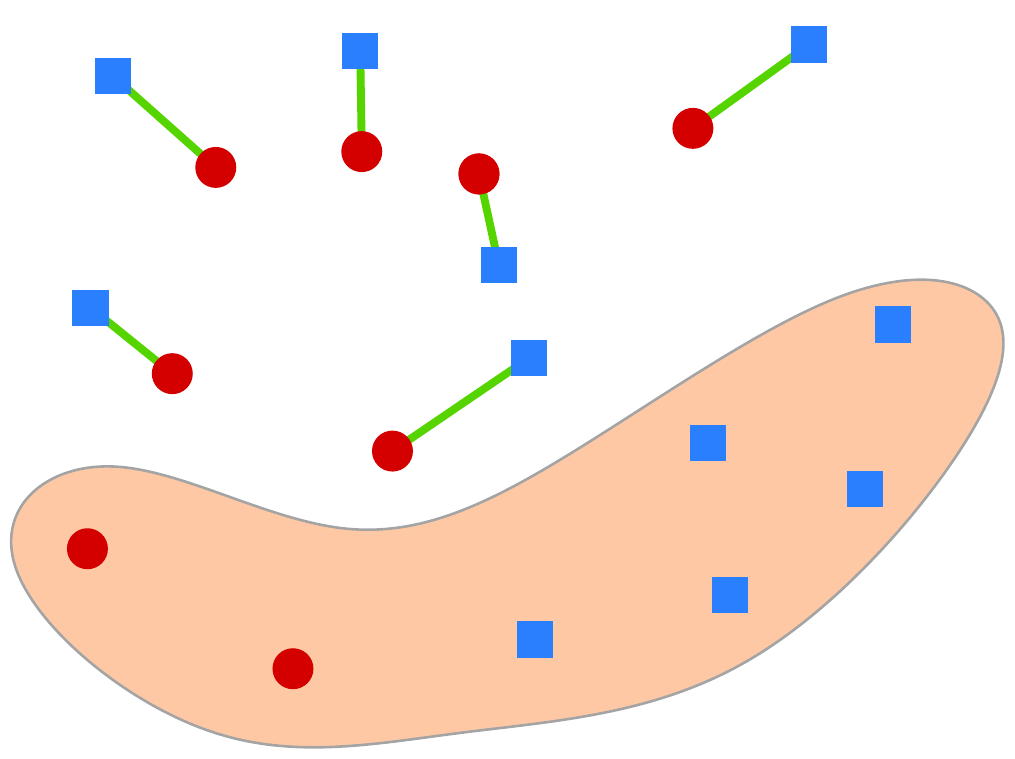}
    \hspace*{\fill}
    \vline\vline
    \hspace*{\fill}
        \includegraphics[width=0.35\textwidth]{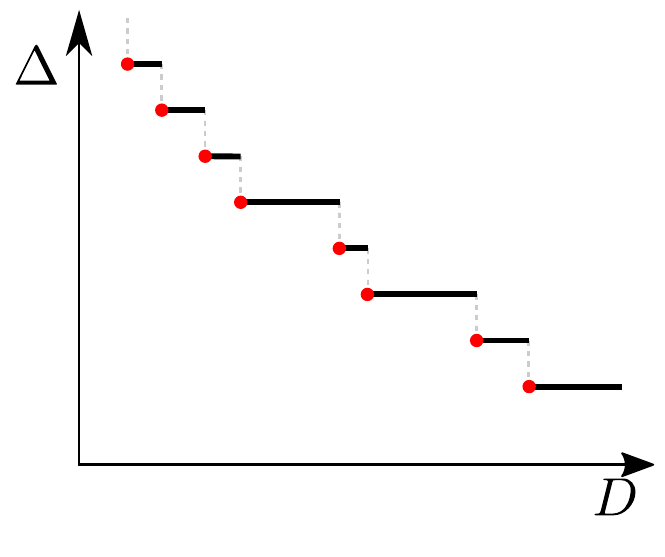}
    \caption{\footnotesize An illustration of metric error with outliers. The figure on the left shows a request set $\reqs$ (red circles) and a prediction $\preds$ (blue squares). The figure in the middle shows a valid definition of metric error with outliers for $\preds$ w.r.t. $\reqs$. Here, $\mreqs\subset \reqs$ are matched to $\mpreds\subset \preds$ using the green lines; $D$ is the total length of the green lines. The outliers are shown in the orange shaded region; $\Delta$ is the number of points in this region. \\
    The figure on the right shows the tradeoff between $\Delta$ and $D$ in the error. The red points are the Pareto frontier formed by various choices of $\Delta$ and $D$.}
    \label{fig:Intro_ViewFigure}
\end{figure}

\smallskip\noindent
{\bf A Framework for Combining Offline and Online Algorithms.}
The gold standard in online algorithms with predictions is 
to interpolate between online and offline bounds for a problem 
as a function of the prediction error. 
But, perhaps surprisingly, in spite of 
the recent surge of interest in this domain, there is still no
general framework/methodology that can combine online and offline 
results in a black box fashion to achieve such interpolation.
I.e., the existing results are individually tailored to meet the 
needs of the problem being considered, and therefore, do not
use offline/online results in a black box manner.
We make progress toward this goal of designing a general methodology
in this paper by presenting a framework that requires the following 
ingredients:
\begin{itemize}
    \item An {\em online} algorithm that satisfies a technical 
    condition called {\em subset competitiveness}. We will describe
    this more precisely later, but intuitively, this means that the 
    cost incurred by the algorithm on any subset of the input can
    be bounded against the optimal solution on that subset. 
    \item An {\em offline} algorithm for the {\em prize-collecting}
    version of the problem, i.e., if one is allowed to not satisfy 
    a constraint in lieu of paying a penalty for it.
\end{itemize}
We show that this framework achieves
tight bounds for several classical problems in graph algorithms
that we describe below. But, before 
describing these results, we note that the framework itself 
is quite general (i.e., does not use any specific property of 
problems defined on graphs), and uses 
the online and offline algorithms as black box results. Therefore,
we hope that this framework will be useful for online problems
with predictions in other domains as well.

\smallskip\noindent
{\bf Our Problems.}
We present online algorithms with predictions for some classic 
problems in graph algorithms using the above framework:
\begin{itemize}
    \item {\em Steiner Tree} (e.g., \cite{KarpinskiZ97,Zelikovsky93,PromelS00,RobinsZ05,ByrkaGRS13}): 
    In this problem, the goal is to 
    obtain the minimum cost subgraph that connects a given set of 
    vertices called {\em terminals}.
    In the online version, the terminals are revealed
    one at a time in each online step. The offline prediction 
    provides the algorithm with a predicted set of terminals.
    \item {\em Steiner Forest} (e.g., \cite{AgrawalKR95,GoemansW95}): 
    This is a generalization of the Steiner
    tree problem where the input comprises a set of vertex pairs 
    called {\em terminal pairs} that
    have to be connected to each other 
    (but not necessarily with other pairs). 
    In the online version, the terminal pairs are revealed
    one at a time in each online step. Again, the offline prediction
    gives a predicted set of terminal pairs.
    \item {\em Facility Location} (e.g.,~\cite{Hochbaum82,ShmoysTA97,Meyerson01,GuhaK99,ChudakS03,KorupoluPR00,JainV01,CharikarG05,JainMS02,JainMMSV03,Sviridenko02,MahdianYZ06,ByrkaA10,Li13}): In this problem, a subset 
    of vertices are designated as {\em facilities} and each of them is 
    given an opening cost. Now, a subset of vertices are identified as
    {\em clients}, and each client must connect to an {\em open} facility
    and pay the shortest path distance to that vertex as connection cost.
    The open facilities also pay their corresponding opening costs.
    The goal is to minimize the total cost. In the online version, the 
    clients are revealed one at a time in each online step. The offline 
    prediction comprises a set of predicted clients.
    \item {\em Capacitated Facility Location} (e.g., \cite{JainV01,DBLP:conf/random/MahdianYZ03}):
    This is the same as facility location, except that each facility
    also has a maximum capacity for the number of clients that can 
    connect to it. We consider the {\em soft} capacitated version 
    of the problem, i.e., where multiple copies of a facility can be 
    opened but each opened copy incurs the opening cost of the facility.
\end{itemize}

\smallskip\noindent
{\bf Our Results.}
Our main result is the following:
\begin{thm}
\label{thm:robust-simple}
    Consider the Steiner tree, Steiner forest, and (capacitated) facility location problems.
    Suppose $\reqs$ and $\preds$ are respectively the actual and the predicted
    inputs, and let $\opt$ be the optimal solution for $\reqs$.
    Then, for {\em\bf any} metric error with outliers $(\Delta, D)$, 
    we present (for each problem) an algorithm $\alg$ with the following guarantee:
    \begin{equation}
        \alg \leq O(\log \Delta) \cdot \opt + O(D).
    \end{equation}
    
\end{thm}

{\bf Universality:} As shown in \cref{fig:Intro_ViewFigure}, the multiple choices for the error $(\Delta,D)$ create a Pareto frontier of choices which are non-dominated. 
The algorithms presented in this paper achieve the bounds in \cref{thm:robust-simple} for \emph{every error $(\Delta,D)$ simultaneously}, and thus remain oblivious to the chosen error parameters.

{\bf Scale-freeness:} The key property of these theorems is that the competitive ratio does not
scale with the size of the input as long as the input is predicted correctly 
(or with small $D$). Conceptually, this is exactly what predictions should
achieve in an online algorithm, that the competitive ratio should only be
a function of the part of the input that was not predicted correctly.

{\bf Tightness of Bounds:} The competitive ratios in Theorem~\ref{thm:robust-simple} match the tight competitive ratios
for the corresponding online problems, i.e., for $\preds = \emptyset$. Hence, these
bounds are also tight. Indeed, we show in the appendix that the online lower 
bound constructions can be implemented even if there
is a prediction of {\em any} size by creating a copy of the lower bound 
instance inside the erroneous predictions.

{\bf Impossibility of Black Box Theorem:} 
Given our theorems, one might wonder if it is {\em always} possible 
to replace the number of requests in an online graph algorithm with the 
number of prediction errors in the competitive ratio, i.e., if there is 
a {\em universal} theorem to this effect. Unfortunately, this is not true. 
For instance, in 
the bipartite matching problem even on a uniform metric space on $n$ points,
it is easy to see that even a single prediction error 
can lead to a competitive ratio of $\Omega(\log n)$. We give this example 
in the appendix.

{\bf Extension to Priority Steiner tree/forest:} 
In \Cref{sec:PSF}, we extend our results to \emph{priority} Steiner tree/forest (e.g. \cite{CharikarNS04,ChuzhoyGNS08}). 
    This is a generalization of Steiner tree/forest where every 
    edge and every terminal/terminal pair has a {\em priority}, 
    a terminal/terminal pair can only use edges of its priority 
    or higher to connect. In the online version, the terminals/terminal pairs
    and their priorities are revealed online. The offline prediction 
    comprises a set of predicted terminals/terminal pairs with their
    respective priorities. In the video conferencing and network design
    applications, this reflects multiple {\em service level agreements}
    (SLAs) with different users.
    
    We give an algorithm $\alg$ with the following guarantee:
        \begin{equation}
            \alg \leq O(\npry\cdot \log (\Delta/\npry)) \cdot \opt + O(D),
        \end{equation}
    where $\npry$ is the number of priority classes in the input.

\section{A General Framework for Online Algorithms with Predictions}
\label{sec:framework}

In this section, we give our general framework for online algorithms with prediction. Abstractly, we are given a universe $\S$ offline, where each element has an associated cost $c:\S\to \mathbb{R}^+$. A set of requests $\reqs$ arrives online, one request per step, and the online algorithm must augment its solution $S\subseteq \S$ by adding new elements from $\S$ such that the new request is satisfied. We assume that for each request $\req\in \reqs$, the collection of sets satisfying $\req$ is {\em upward-closed}, i.e., if $S$ satisfies $\req$, so too does any $S'\supseteq S$. 
The algorithm is also given an offline prediction $\preds$ for the set of requests. 


For every set $S\subseteq \S$ of elements, we define $c(S):= \sum_{e\in S} c(e)$. 
Our framework requires the following ingredients:
\begin{property}
	\label{asmp:FW_SubsetCompetitive}
	An online algorithm $\ON$ whose competitive ratio is  $f(|\reqs|)$. This algorithm must be \emph{subset-competitive}; that is, for every request subset $\reqs' \subseteq \reqs$ we have 
	\[ \ON(\reqs') \le O(f(|\reqs'|)) \cdot \OPT \]
	where $\ON(\reqs')$ denotes the total cost incurred by the algorithm in serving the requests of $\reqs'$.
\end{property}

\begin{property}
	\label{asmp:FW_PrizeCollecting}
	 A (constant) $\gamma$-approximation algorithm $\PC$ for the prize-collecting offline problem.
In this problem, each request in $\reqs$ has a penalty given by $\pi:\reqs \to \mathbb{R}^+$. A solution is a subset of elements $S\subseteq \S$ that minimizes $c(S) + \sum_{\req\in \reqs'} \pi(\req)$, where $\reqs' \subseteq \reqs$ is the set of requests \underline{not} satisfied by $S$.
\end{property}




Given these two ingredients, we define our framework in Algorithm~\ref{alg:FW_Algorithm}. Before describing the framework formally, let us give some intuition for it. There are two extreme cases. If the prediction is entirely accurate, we should simply run an offline approximation for the predicted requests. On the other hand, if the prediction is entirely erroneous, we should run a standard online algorithm and ignore the prediction completely. Between these extremes, where the prediction is only partially accurate, some hybrid of those algorithms seems a natural choice. 
Thus, our framework runs a standard online algorithm on the actual requests, while also satisfying increasingly larger parts of the prediction by periodically running an offline algorithm to satisfy predicted requests. 
Ideally, the offline solutions should satisfy the correctly predicted requests, while the online solution should satisfy the requests that were not predicted.
But, in general, this is impossible to achieve: the offline algorithm does not know in advance which of the predicted requests will arrive in the future. Our main technical work is to show the surprising property that it suffices for the offline algorithm to choose the most ``cost-effective'' part of the prediction to satisfy {\em even if those requests eventually do not realize in the actual input}.

\newcommand{\onc}{B}
\newcommand{\offc}{\hat{B}}

\DontPrintSemicolon
\begin{algorithm}[t]
	\caption{\label{alg:FW_Algorithm}General Framework for Online Algorithms with Predictions}

	\algorithmfootnote{\footnotemark[\value{footnote}] This can be done by enumerating all values of $i$.
	}
	\begin{footnotesize}
	\Input{$\ON$ -- An online algorithm for the problem \\
	$\PC$ -- a $\gamma$-approximation for the prize-collecting problem \\
	$\preds$ -- a prediction of requests}

	\BlankLine

	\Initialization{}{
		Initialize $\onc\gets0$, $\offc\gets0$ and $S\gets\emptyset$.
		%

		Initialize $\ON$ to be a new instance of the online algorithm.
	}
	\BlankLine

	\EFn(\tcp*[h]{Upon the next request $r$}){\UponRequest{$r$}}{

		Send $\req$ to the online algorithm $\ON$, and augment $S$ accordingly. Increase $\onc$ by the resulting cost incurred by $\ON$.\label{line:FW_PathConnect}

		\BlankLine

		\tcp*[h]{Whenever the cost of the online algorithm doubles, buy another offline solution}

		\If{$\onc \ge 2\offc$}{\label{line:FW_IfCondition}

			Set $\offc\gets \onc$

			Calculate $0 \le u \le |\preds|$, the minimum number such that $c(\Partial(\preds,u))\le 3\gamma \offc$. \label{line:FW_DefiningHatK}

			Augment $S$ with the elements of $\Partial(\preds,u)$. \label{line:FW_AddPartial}

			Start a new instance of the online algorithm $\ON$ where $c(e) = 0$ for all $e\in S$. \label{line:FW_RestartOnline}

		}
	}

	\BlankLine

	\Fn{\Partial{$\preds$,$u$}}{

		\lIf{$ \gamma u \ge |\preds| $}{\Return{$ \emptyset $}}

		For every $x$, define $\pi_x$ to be the prize-collecting penalty function that penalizes every request by $x$.

		Let $i$ be the minimum integer such that $\PC (\preds, \pi_{2^i})$ does not satisfy $\le \gamma u$ requests.\footnotemark{}

		Let $S_1 =\PC(\preds,\pi_{2^{i-1}})$ and $S_2 = \PC (\preds, \pi_{2^i})$.

		Let $u_1,u_2$ be the number of requests from $\preds$ which are not satisfied by $S_1,S_2$, respectively.

		\leIf{$\gamma u \ge \frac{u_1+u_2}{2}$ }{\Return{$S_1$}}{ \Return{$S_2$} }

	}
	\end{footnotesize}
\end{algorithm}

\textbf{Framework and \Partial Procedure.} A high-level, informal description of the framework is to perform the following whenever a request $\req$ is released:

\begin{enumerate}
    \item Send $\req$ to be served by the online algorithm $\ON$.
    
    \item If the total cost $\onc$ incurred in Step 1 (over all requests) has doubled, buy an offline solution which satisfies as many requests from the prediction $\preds$ as possible, while not exceeding a budget of $O(\onc)$.
\end{enumerate}

The framework thus runs the online algorithm $\ON$ on each incoming request, but periodically adds an (offline) solution for the prediction $\preds$. This offline solution is computed by a subroutine that we call \Partial. The input to $\Partial$ is the prediction $\preds$ and an integer $u$ such that $0\le u \le |\preds|$, and it attempts to output a minimum-cost solution that satisfies all requests from $\preds$ {\bf except} for at most $u$ requests. If $\Partial$ were able to solve this minimization problem exactly, or even up to a constant approximation, then the algorithm would add the solution returned by $\Partial(\preds,u)$ for the minimum value of $u$ within the given budget. An example for a problem in which $\Partial$ can output such a solution is Steiner tree, in which $\Partial$ simply needs to solve the well-known $k$-\mst problem, for which there are several constant approximation algorithms (e.g.,~\cite{AwerbuchABV98,BlumRV99,Garg96,AryaR98,AroraK06,Garg05}).  

However, for other problems, this may be infeasible -- for Steiner forest, the approximability of the corresponding ``$k$-minimum spanning forest'' problem is known to be related to that of the $k$-densest subgraph problem, and hence a sub-polynomial approximation is unlikely~\cite{HajiaghayiJ06,SegevS06,GuptaHNR10}.

Our key observation is that we can relax the guarantee that we seek from this minimization problem. Namely, we show that it suffices to obtain a {\em residual bi-criteria} constant approximation, which means that the algorithm satisfies the following two properties: (a) the cost of the solution is at most a constant times that of the optimal solution, and (b) the number of requests in $\preds$ that are {\em not satisfied} by the algorithm is also at most a constant times that in the optimal solution. Surprisingly, this weaker guarantee allows us to entirely conceal the differences in approximability of the minimization problem being solved by $\Partial$. Indeed, we show that using a series of invocations of an offline prize-collecting algorithm $\PC$ in a completely {\em black box} manner, we can obtain an implementation of $\Partial$ that provides the bicriteria guarantee that we need.

The formal algorithm for the framework is presented in Algorithm~\ref{alg:FW_Algorithm}.

\begin{figure}
    \begin{center}
    \includegraphics[width=0.75\textwidth]{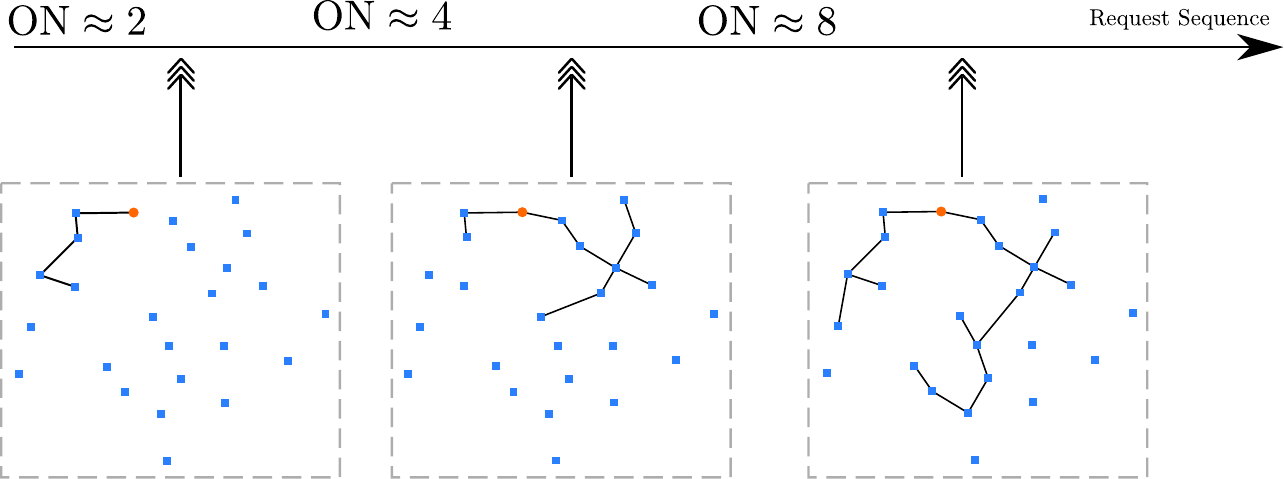}
    \end{center}
    \caption{\footnotesize Visualization of Framework: As the request sequence advances, the requests are forwarded to an online algorithm. Each time the cost of this online algorithm doubles, the current solution is augmented with an offline solution to the prediction, which serves an increasingly-larger number of predicted requests.}
    \label{fig:FW_Timeline}
\end{figure}

\subsection{Framework Properties}

\paragraph{The \text{\sc Partial}\xspace Subroutine Guarantee.} Recall the desired property from the subroutine $\Partial$: given a prediction $\preds$ and a number $u$ (such that $0\le u\le |\preds|$) as input, compare the output $S$ of $\Partial(\preds,u)$ against the minimum-cost solution $S^*$ which satisfies all requests from $\preds$ except for $u$ requests. This output $S$ cannot fail to satisfy more than a constant factor more requests $\preds$ than $S^*$; that is, it fails to satisfy only $O(u)$ requests (here, recall that $\gamma$ is a constant). In addition, the cost of $S$ can be at most a constant times the cost of $S^*$. This residual bi-criteria guarantee crucially uses Property \ref{asmp:FW_PrizeCollecting}, where the $\gamma$ approximation ratio of the prize-collecting algorithm goes into both the cost of $S$ and the amount of requests satisfied by $S$. This guarantee is formally stated in the following lemma (proof in Appendix \ref{sec:framework_AdditionalProofs}).


\begin{lem}
	\label{lem:FW_Partial}
	The solution $S$ returned by a call to $\Partial(\preds,u)$ has the following properties:
	\begin{enumerate}
		\item $S$ satisfies all requests from $\preds$ except for at most $2\gamma u$ requests. 
		\item $ c(S) \le 3\gamma\cdot  c(S^\ast)$, where $S^\ast$ is the minimum cost solution satisfying all requests from $\preds$ except for at most $u$ requests.
	\end{enumerate}
\end{lem}

\paragraph{Bounding Cost of Framework.} Next, we provide some bounds for various costs incurred by the framework. The following sections make use of these general bounds to prove competitiveness.

The requests of $\reqs$ are handled in $\pa{\reqs}$ iterations of the algorithm (calls to \UponRequest). For the variable $\onc$ in the algorithm, and any iteration $j$, we write $\onc_j$ to refer to the value of the variable immediately after the $j$'th iteration of the algorithm (where $\onc_0$ is the initial value of $\onc$, namely $0$). 
Similarly, we define $\offc_j$ to be the value of the variable $\offc$ after the $j$ iteration, where $\offc_0 = 0$.



If the \textbf{if} condition of Line \ref{line:FW_IfCondition} holds during iteration $i$, then $i$ is called a \emph{major iteration}. That is, major iterations are those iterations in which our solution is augmented by a solution for some predicted requests.

Fix a major iteration $i$, and observe the subsequent major iterations, denoted $i= i_0^{\star}, i_{1}^{\star},...,i_{m}^{\star}$ (where $m$ is the number of such major iterations). For convenience, we define $i_{m+1}^\star = |R|$. For each index $j \in \{0,\cdots,m\}$, we define  \emph{phase} $j$ to consist of iterations $\left\{ i_j^{\star}+1 , \cdots, i_{j+1}^{\star} \right\}$. Observe that each phase uses a different instance of $\ON$; we denote by $\ON_j$ the total cost of the $\ON$ instance of phase $j$.

\begin{figure}
    \begin{center}
    \includegraphics[width=0.5\textwidth]{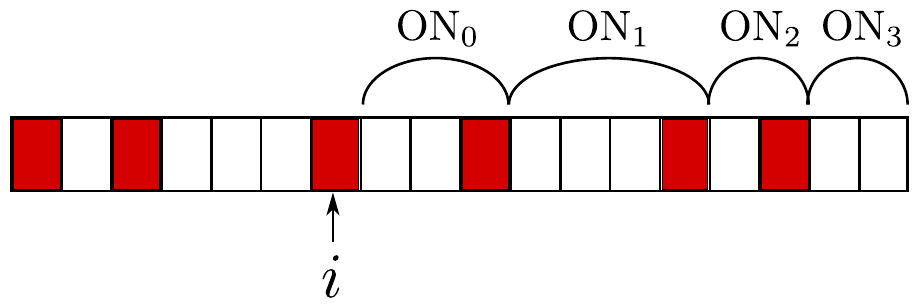}
    \end{center}
    \caption{\footnotesize Phase Structure: This figure illustrates the iterations of the algorithm, where the major iterations are shown in red. Fixing any specific major iteration $i$ induces a sequence of online phases $\ON_0, \cdots, \ON_m$.}
    \label{fig:FW_PhaseStructure}
\end{figure}

The following lemma provides a bound on the cost of the algorithm (formal proof in Appendix \ref{sec:framework_AdditionalProofs}). It does so by dividing the costs of the algorithm into two parts, subject to the choice of any major iteration $i$: the first is the prefix cost, up to (and including) iteration $i$; the second is the suffix cost, from iteration $i+1$ onwards. Intuitively, a ``correct'' choice of $i$ would be such that $\offc_{i}\approx \OPT$; subject to such a choice of $i$, the first $i$ iterations perform an exploratory doubling procedure, and thus the prefix cost would be $O(\OPT)$. As for the suffix cost, we claim that once an offline solution of cost $\approx \OPT$ is added, the suffix cost cannot be large unless the prediction is rather inaccurate. These claims are made formal in the following sections; for now, we simply state and prove this prefix-suffix division of costs. Note that this lemma makes use of Property \ref{asmp:FW_SubsetCompetitive}, which implies that serving no single request in the online algorithm is prohibitively expensive.


\begin{lem}
    \label{lem:FW_BoundingFrameworkCost}
    Fix any major iteration $i$. Denote the subsequent major iterations by $i= i_0^{\star}, i_{1}^{\star},...,i_{m}^{\star}$, and define phases and $\ON_j$ for any $j$ as above (with respect to iteration $i$). The following holds:
    \begin{enumerate}
        \item The total cost of the first $i$ iterations is at most $O(1)\cdot \OPT + (12\gamma + 2) \cdot \offc_{i-1}$.
        
        \item The total cost of the iterations $i+1,\cdots, |\reqs|$ is at most $O(1)\cdot \max\pc{ \ON_{m-1},\ON_m } $.
    \end{enumerate}
\end{lem}

The lemmas of this section reduce the cost analysis of the framework to bounding the terms $\ON_{m-1},\ON_{m}$ and  $\offc_{i-1}$ for some major iteration $i$. This is done in the following sections, for each problem separately.

\section{Online Steiner Tree with Predictions}
\label{sec:ST}

In the Steiner tree problem, we are given an undirected graph $G=(V,E)$ with edge costs $c:E\to \mathbb{R}^+$. The set of requests $\reqs=\{\req_1, \req_2, \ldots, r_{|R|}\}$ is a subset of vertices called {\em terminals}. The goal is to obtain a connected subset of edges $S$ that spans all the terminals at minimum cost $\sum_{e\in S} c(e)$. In the {\em online} version, the terminals arrive online, one per step, and must be connected to a root vertex $\rho$ given offline.
(The rooted and unrooted versions are equivalent.)
The offline prediction comprises a predicted set of terminals $\preds$.

In this section, we apply the framework of Algorithm \ref{alg:FW_Algorithm} using the two required components -- an offline algorithm for prize-collecting Steiner tree and a subset-competitive online algorithm for Steiner tree.

\smallskip\noindent{\bf Prize-collecting offline algorithm.}
Goemans and Williamson \cite{GoemansW95} gave a 2-approximation algorithm for the prize-collecting (unrooted) Steiner tree problem. Their algorithm can be trivially adapted for the rooted version of the problem by adding the root $\rho$ to the terminals, and assigning it a penalty of $\infty$.
\begin{lem}[Goemans and Williamson \cite{GoemansW95}]
	\label{fact:ST_PrizeCollecting}
	There exists a 2-approximation offline algorithm $\PCST$ for the rooted prize-collecting Steiner tree problem.
\end{lem}

{\bf Subset-competitive online algorithm.}
The greedy algorithm $\OST$ for online Steiner tree~\cite{ImaseW91} performs the following:
\begin{enumerate}
    \item Upon the release of a new terminal $r_i$, let $v \in \pc{\rho, \req_1,\cdots \req_{i-1}}$ be the closest terminal to $\req_i$.
    
    \item Buy the path $P_{r_i,v}$, the shortest path from $\req_i$ to $v$.
\end{enumerate}

As shown in~\cite{ImaseW91}, this algorithm is $O(\log |\reqs|+2)$-competitive.
The following lemma shows that it also upholds Property \ref{asmp:FW_SubsetCompetitive}, i.e. it is subset competitive.

\begin{lem}
	\label{lem:ST_Greedy_SubsetCompetitive}
	The greedy algorithm $\OST$ is subset-competitive, as in Property \ref{asmp:FW_SubsetCompetitive}. That is, for an input $\reqs$ and every $\reqs' \subseteq \reqs$, it holds that 
	\[
	 \OST(\reqs') \le  O \pr{\log \pa{\reqs'}+2} \cdot \OPT
	\] 
\end{lem}
\begin{proof}
    Denote by $\OPT'$ the optimal solution for the subset $\reqs'$. We show that $\OST(\reqs') \le O\pr{\log \pa{\reqs'}+2} \cdot \OPT'$, which implies the lemma since $\OPT' \le \OPT$.
    
    Note that for each terminal $\req\in \reqs'$,
    the cheapest path to a previous terminal in an instance with terminals $\reqs$ is at most
    as expensive as the cheapest path to a previous terminal in an instance with terminals $\reqs'\subseteq \reqs$.
    The lemma now follows from the competitive ratio of the instance with terminals $\reqs'$.
\end{proof}

\paragraph{Steiner tree with predictions.} The algorithm for Steiner tree with predictions is now constructed by simply plugging in $\OST$ and $\PCST$ into the framework.
Note that the solution $S$ remains connected throughout the algorithm (specifically in Line \ref{line:FW_AddPartial} of Algorithm \ref{alg:FW_Algorithm}, both subgraphs whose union is taken contain the root vertex $\rho$).

We would now like to show that applying the framework to $\OST$ and $\PCST$ yields a guarantee for metric errors with outliers. However, we must first define the matching cost of two requests, to complete the definition of errors in the Steiner tree problem. For Steiner tree, we use the natural choice that the matching cost of two requests $\req_1,\req_2$ is simply the distance between them in the graph, henceforth denoted by $d(\req_1,\req_2)$.

Fix any prediction $\preds$ and any input request sequence $\reqs=(\req_1,\dots,\req_{|\reqs|})$. Let $\err = (\Delta,D)$ be some error for the input $\reqs$ and the prediction $\preds$. We denote by $\mreqs[\err]\subseteq \reqs$ and $\mpreds[\err]\subseteq \preds$ the subsets of matched input requests and matched predicted requests, respectively, in the definition of $\err$. Slightly abusing notation, we use $\err$ as the matching between $\mreqs[\err]$ and $\mpreds[\err]$ such that for every $\req\in \mreqs[\err]$, we use $\err(\req)$ to denote its matched predicted request in $\mpreds[\err]$. Similarly, for a subset of requests $\reqs' \subseteq \mreqs[\err]$, we define $\err(\reqs') \subseteq \mpreds[\err]$ to be the set of predicted requests matched to requests in $\reqs'$. In the other direction, we use $\err(\pred)$ to denote the request in $\mreqs[\err]$ matched to the predicted request $\pred \in \mpreds[\err]$ (operating similarly for sets of predicted requests). Note that by the definition of $(\Delta,D)$-error, it holds that $\Delta = \pa{\reqs\backslash\mreqs[\err]}+\pa{\preds\backslash \mpreds[\err]}$, and that $D=\sum_{\req \in \mreqs[\err]} d(\req, \err(\req))$.

The error guarantee for Steiner tree is formally stated in Theorem \ref{thm:RST_Competitiveness}.

\begin{thm}
	\label{thm:RST_Competitiveness}
	For every error $\err = (\Delta,D)$, Algorithm \ref{alg:FW_Algorithm} for Steiner tree with predictions upholds
	\[ 
	    \ALG \le O\pr{ \log\pr{\min\pc{\pa{\reqs},\Delta}+2} } \cdot \OPT +O(1)\cdot D.   
	\]
\end{thm}

Henceforth, we fix the error $\err$, and drop $\err$ from $\mreqs[\err]$ and $\mpreds[\err]$. We define $k = \pa{\mreqs[\err]}$ the number of matched requests in $\err$, and fix $i$ to be the major iteration in which $u$ is first assigned a value which is at most $\pa{\preds} - k$ in Line \ref{line:FW_DefiningHatK} of Algorithm \ref{alg:FW_Algorithm}.

\begin{proof}[Proof of Theorem \ref{thm:RST_Competitiveness}]
    Having used the framework of Section \ref{sec:framework}, and observing that $i$ is a major iteration, we can apply Lemma \ref{lem:FW_BoundingFrameworkCost}:
    \begin{enumerate}
        \item The cost of the first $i$ iterations is at most $O(1)\cdot \OPT + (12\gamma + 2)\offc_{i-1}$. Using Proposition \ref{prop:RST_OPTAndDistAtLeastPartial}, this expression is at most $O(1)\cdot \OPT + O(1)\cdot D$.
        
        \item The cost of the last iterations, from $i+1$ onwards, is at most $O(1)\cdot \max\pc{\OST_{m-1},\OST_{m}}$. Using Lemma \ref{lem:RST_BoundingExpensiveON}, this is at most $O\pr{\log\pr{\min\pc{|\reqs|,\Delta}}} \cdot\OPT + O(1) \cdot D$.
    \end{enumerate}
    
    Overall, we have that 
    \[
        \ALG \le O\pr{\log\pr{\min\pc{|\reqs|,\Delta}}} \cdot\OPT + O(1) \cdot D.  
    \]
\end{proof}

To complete the proof of Theorem \ref{thm:RST_Competitiveness}, it remains to prove Proposition \ref{prop:RST_OPTAndDistAtLeastPartial} and Lemma \ref{lem:RST_BoundingExpensiveON}.

\begin{prop}
    \label{prop:RST_OPTAndDistAtLeastPartial}
	$\offc_{i-1}\le \OPT + D  $.
\end{prop}
\begin{proof}
    Consider the iteration $i'<i$ in which $\offc_{i-1}$
	was set, i.e. the first iteration such that $\offc_{i'}=\offc_{i-1}$. If $i'=0$, then $\offc_{i-1}$ is the initial value of $\offc$, which is $0$, and the proposition trivially holds. Henceforth assume that $i'>0$. 
	
	From the definition of $i$, we have that $u_{i'}>|\preds|-k$, where  $u_{i'}$ is the value of the variable $u$ after iteration $i'$. Thus, $c(\Partial(\preds,|\preds|-k))>3\gamma \offc_{i^{\prime}}=3\gamma \offc_{i-1}$. From Lemma \ref{lem:FW_Partial}, this implies that the least expensive solution which satisfies at least $k$ requests from $\preds$ costs at least $\offc_{i-1}$. Since $k = |\mreqs|=|\mpreds|$, we have that $\OPT_{\mpreds}$ is such a solution, where $ \OPT_{\mpreds}  $ is the optimal solution for $\mpreds$. Therefore, $c(\OPT_{\mpreds}) \ge \offc_{i-1}$.
	
	Now, observe that $\OPT$ can be augmented to a solution for $\mpreds$ through connecting $\req$ to $\err(\req)$ for every $\req \in \mreqs$, at a cost of at most $D$. Thus, $c(\OPT) + D \ge c(\OPT_{\mpreds}) \ge \offc_{i-1}$, which completes the proof.
    %
\end{proof}


Define the phases of the algorithm with respect to the major iteration $i$ as in Section \ref{sec:framework}, where the cost of the online instance $\OST$ in the $j$'th phase is denoted by $\OST_j$. We again number the phases from $0$ to $m$.

\begin{lem}
	\label{lem:RST_BoundingExpensiveON}
	It holds that 
	\[
	    \max\pc{\OST_{m-1}, \OST_{m}}\le O\pr{\log \pr{\min\pc{|\reqs|,\Delta}+2}} \cdot\OPT + O(1) \cdot D.
	\]
\end{lem}
\begin{proof}
    Let $j\in \pc{m-1,m}$, and let $Q\subseteq \reqs$ be the subsequence of requests considered in $\OST_j$. Denote by $\preds'\subseteq \preds$ the subset of predicted requests that are satisfied by the $\Partial$ solution bought in iteration $i$. 
    
    The online algorithm $\OST_j$ operates in a modified metric space, in which the cost of a set of edges $S_0$ is set to $0$; denote by $d'$ the metric subject to these modified costs. 
    
    We partition $Q$ into the following subsequences:
    \begin{enumerate}
        \item $Q_1 = Q \cap \mreqs$. We further partition $Q_1$ into the following sets:
        \begin{enumerate}
            \item $Q_{1,1} = \pc{r\in Q_1 | \err(r) \in \preds'}$
            \item $Q_{1,2} = \pc{r\in Q_1 | \err(r) \notin  \preds\backslash \preds'}$
        \end{enumerate}
        
        \item $Q_2 =  Q\backslash \mreqs$.
    \end{enumerate}
    
    Observe that $\OST_j(Q) = \OST_j(Q_{1,1}) + \OST_j(Q_{1,2} \cup Q_2)$. We now bound each component separately.
    
    \paragraph{Bounding $\OST_j(Q_{1,1})$.} Since $\err(Q_{1,1}) \subseteq \preds'$, and $\preds'$ is already served by the $\Partial$ solution bought in iteration $i$, it holds for each request $\req\in Q_{1,1}$ that $d'(\err(\req),\rho)=0$. Through triangle inequality, it thus holds that $d'(\req,\rho) \le d'(\req,\err(\req))$, and thus the cost of connecting $\req$ incurred by $\OST_j$ is at most $d'(\req,\err(\req))$. Therefore:
    \[ \OST_j(Q_{1,1}) \le \sum_{\req\in Q_{1,1}} d'(\req,\err(\req)) \le \sum_{\req\in Q_{1,1}} d(\req,\err(\req)) \le D \] 
    
    \paragraph{Bounding $\OST_j(Q_{1,2} \cup Q_2)$.} Using Property~\ref{asmp:FW_SubsetCompetitive} (subset competitiveness), we have that 
    
	\[ 
	    \OST_j(Q_{1,2} \cup Q_2) \le O\pr{\log \pr{\pa{Q_{1,2}\cup Q_2}+2}} \cdot \OPT'  
	\]
	where $\OPT'$ is the optimal solution to $Q$ with the cost of $S_0$ set to $0$. Clearly, $\OPT' \le \OPT$. 
	
    Now, observe that $\pa{Q_{1,2}} \le \pa{\preds\backslash \preds'}$. From the definition of the major iteration $i$, and from Lemma \ref{lem:FW_Partial}, it holds that $\pa{\preds\backslash \preds'} \le 2\gamma_{\ST} \cdot (|\preds| - k) = 2\gamma_{\ST} \pa{\preds\backslash \mpreds}$. As for $Q_2$, it holds that $Q_2 \le \pa{\reqs \backslash \mreqs}$. These facts imply that 
    \[ 
        \pa{Q_{1,2}\cup Q_2} \le 2\gamma_{\ST} \Delta
    \]
    In addition, it clearly holds that $\pa{Q_{1,2}\cup Q_2} \le \pa{\reqs}$. Therefore, it holds that 
    \[  
        \OST_j(Q_{1,2}\cup Q_2) \le O\pr{\log \pr{\min\pc{\pa{\reqs},\Delta}+2}}\cdot \OPT
    \]
    
    Combining this with the previous bound for $\OST_j(Q_{1,1})$, we obtain
    \[ 
        \OST_j(Q) \le O\pr{\log \pr{\min\pc{\pa{\reqs},\Delta}+2}}\cdot \OPT + D
    \] 
    completing the proof.
\end{proof}

\eat{
\begin{proof}[Proof of Theorem \ref{thm:RST_Competitiveness}]
    We choose iteration $i$, which is a major iteration, and apply Lemma \ref{lem:FW_BoundingFrameworkCost}.
    
    Combining Proposition \ref{prop:RST_OPTAndDistAtLeastPartial} with the first item of Lemma \ref{lem:FW_BoundingFrameworkCost}, we have that the cost of the first $i$ iterations is at most $O(1)\cdot \OPT + O(1)\cdot D$. 
    
    Combining Proposition \ref{lem:RST_BoundingExpensiveON} with the second item of Lemma \ref{lem:FW_BoundingFrameworkCost}, we have that the cost of iterations $i+1$ onwards is at most 
    \[ 
        O\pr{\log\pr{\min\pc{|\reqs|,\Delta}}} \cdot\OPT + O(1) \cdot D
    \]
    
    Summing all costs of the algorithm, we have that     
    \[  
        \ALG \le O\pr{\log\pr{\min\pc{|\reqs|,\Delta}}} \cdot\OPT + O(1) \cdot D 
    \]
    as required.
\end{proof}
}

\section{Online Steiner Forest with Predictions}
\label{sec:SF}
In the online Steiner forest problem, we are given an undirected graph $G=(V,E)$ with non-negative edge costs $c:E\to \mathbb{R}^+$. The requests are pairs of vertices from $V$ called {\em terminal pairs}. The algorithm maintains a subgraph $S$ of $G$ in which each terminal pair is connected. Upon the arrival of a new terminal pair, the algorithm must connect the two vertices by adding edges to $S$. 
The goal is to minimize the cost of the solution $c(S)$. We are also given a set of vertex pairs offline as the prediction $\preds$ for the terminal pairs.

In this section, we apply the framework in Algorithm \ref{alg:FW_Algorithm} to Steiner forest, using the two components: an offline approximation for prize-collecting Steiner forest, and an online algorithm for Steiner forest.

{\bf Prize-collecting offline algorithm.}
The following algorithm is due to Hajiaghayi and Jain~\cite{HajiaghayiJ06}.
\begin{lem}[Hajiaghayi and Jain~\cite{HajiaghayiJ06}]
	\label{fact:SF_PrizeCollecting}
	There exists a $\gamma_{\SF}$-approximation algorithm $\PCSF$ for prize-collecting Steiner forest, where $\gamma_{\SF}=2.54$.
\end{lem}

{\bf Subset-competitive online algorithm.}
We now describe a version of the algorithm of Berman and Coulston \cite{BermanC97} for online Steiner forest, denoted $\OSF$. The precise statement of the Berman-Coulston algorithm is important, as we want a variant which is subset-competitive, which does not hold for all variants; for example, the variant of Umboh~\cite{Umboh2015} is not subset-competitive. The crucial difference between the version we describe below and that of \cite{Umboh2015} is that when we connect a pair and observe that the balls around the two terminals intersect with existing balls, we connect each terminal to an {\em arbitrary} intersecting ball, while the algorithm in \cite{Umboh2015} connects to {\em all} intersecting balls. This may lead to a high cost at a single step, which invalidates subset competitiveness.

The $\OSF$ algorithm maintains a collection of duals\footnote{Although these collections of disjoint balls actually correspond to feasible LP dual solutions for Steiner forest, and hence we are calling them duals, we will give a self-contained proof of our theorem that does not need the LP formulation or LP duality.} $\left\{D_j\right\}_{j=-\infty} ^\infty$, where each $D_j$ is a set of disjoint balls of radius exactly $2^{j-2}$ in the metric space induced by the graph. For each ball $B\in D_j$, there exists a terminal pair $(s,t)$ such that either $B$ is centered at $s$ and $t$ is outside $B$, or vice-versa. For the sake of analysis, we also define, for each $D_j$, a meta-graph $M_j$ whose vertices correspond to the balls in $D_j$.

The $\OSF$ algorithm is given in Algorithm \ref{alg:SF_BermanCoulston}. Here, $P_{s,t}^F$ denotes the shortest path connecting vertices $s$ and $t$ in $G$, where edges in $F$ have cost $0$. Also, $B(v,r)$ denotes the (open) ball of radius $r$ centered at vertex $v$.

\LinesNumbered \RestyleAlgo{boxruled}\DontPrintSemicolon

\begin{algorithm}
	\algorithmfootnote{\footnotemark[\value{footnote}] Two (open) balls $B(v_1,r_1)$ and $B(v_2,r_2)$ are said to intersect if the shortest path from $v_1$ to $v_2$ costs less than $r_1 + r_2$.}
	\begin{footnotesize}
		\caption{\label{alg:SF_BermanCoulston}Berman and Coulston's Algorithm for Online Steiner Forest ($\OSF$)}

		\Input{$G=(V,E)$ -- the input graph}

		\BlankLine

		\Initialization{}{
			Initialize $D_{j}\gets \emptyset$ for every $j\in\mathbb{Z}$.

			Initialize $F$ to be the empty graph.
		}

		\BlankLine

		\EFn(\tcp*[h]{Upon the next terminal pair $(s,t)$ in the input}){\UponRequest{$(s,t)$}}{

			Set $F\gets F\cup P_{s,t}^F$ and $\ell\gets\lfloor\log c(P_{s,t}^F)\rfloor$.\label{line:SF_BC_ConnectRequest}

			\If{$B(v,2^{\ell -2})$ does not intersect\footnotemark{} any ball in $D_\ell$ for some $v\in \{s,t\}$}
			{
				Add $B(v,2^{\ell -2})$ to $D_\ell$ \tcp*[h]{This adds a vertex to $M_j$}
			}
			\Else(\tcp*[h]{both $B(s,2^{\ell -2})$ and $B(t,2^{\ell -2})$ intersected balls in $D_\ell$})
			{
				Denote by $B_1 = B(a_1,2^{\ell-2})$ some ball that intersected $B(s,2^{\ell -2})$, and
			by $B_2 = B(a_2,2^{\ell-2})$ some ball that intersected $B(t,2^{\ell -2})$.

			\tcp*[h]{Add an edge between $B_1$ and $B_2$ in $M_j$}

			\label{line:SF_BC_ExtraConnection}Set $F \gets F \cup P_{a_1,s}^F \cup P_{a_2,t}^F$.

			}

		}
	\end{footnotesize}
\end{algorithm}

The $O(\log \pa{\reqs})$-competitiveness of Algorithm \ref{alg:SF_BermanCoulston} is due to Berman and Coulston~\cite{BermanC97}. We now show that it also upholds Property \ref{asmp:FW_SubsetCompetitive}, i.e. it is subset competitive. This fact is stated in the following lemma, the proof of which is presented in Subsection \ref{subsec:SF_BC_SubsetCompetitiveness}.

\begin{lem}
	\label{lem:SF_BermanCoulstonSubsetCompetitive}
	Algorithm \ref{alg:SF_BermanCoulston} is subset-competitive as described in Property \ref{asmp:FW_SubsetCompetitive}. That is, for input $\reqs$ and every subset $\reqs'\subseteq \reqs$, it holds that
	\[
	    \OSF(\reqs') \le O \pr{\log \pr{\pa{\reqs'}+2}} \cdot \OPT. 
	\] 
\end{lem}

The algorithm for the online Steiner forest with predictions is now constructed by simply plugging in $\OSF$ and $\PCSF$ into the framework.

\paragraph{Steiner forest with predictions.} Using $\OSF$ and $\PCSF$ as the online and prize-collecting components of the framework in Section \ref{sec:framework}, we obtain the desired algorithm for Steiner forest with predictions. It remains to analyze its guarantees given any error $\err = (\Delta,D)$.

First, we must define the matching cost of requests and predictions for the error to be fully defined. We naturally define the matching cost of two terminal pairs to be the minimum distance one must move the two terminals of the first pair in the metric space in order for each terminal in the first pair to coincide with a distinct terminal in the second pair. Formally, the matching cost of request $(s_1,t_1)$ and request $(s_2,t_2)$ is
\[ 
    \min\pc{d\pr{s_{1},s_{2}} + d\pr{t_{1},t_{2}} , d\pr{s_{1},t_{2}} + d\pr{s_{2},t_{1}}}. 
\]

Fix any prediction $\preds$ and any input request sequence $\reqs=(\req_1,\dots,\req_{|\reqs|})$. As in the Steiner tree case, let $\err = (\Delta,D)$ be some error for the input $\reqs$ and the prediction $\preds$. We again denote by $\mreqs[\err]\subseteq \reqs$ and $\mpreds[\err]\subseteq \preds$ the subsets of matched input requests and matched predicted requests, respectively, in the definition of $\err$. As before, we use $\err$ to map from a request in $\mreqs$ to its match in $\mpreds$ and vice-versa.


The main result of this section is the following theorem for Steiner forest with predictions.

\begin{thm}
	\label{thm:RSF_Competitiveness}For every error $\err = (\Delta, D)$, Algorithm
	\ref{alg:FW_Algorithm} for Steiner forest with predictions has the guarantee
	\[
	    \ALG \le O\pr{ \log \pr{\min\pc{\reqs,\Delta}+2}}  \cdot \OPT +O(1)\cdot D.   
	\]
\end{thm}

As in the Steiner tree case, we henceforth fix a $(\Delta,D)$-error  $\err$, and drop $\err$ from $\mreqs[\err]$ and $\mpreds[\err]$. We define $k = |\mreqs|$. Now, fix $i$ to be the major iteration in which $u$ is first assigned a value which is at most $|\preds|-k$ in Line \ref{line:FW_DefiningHatK} of Algorithm \ref{alg:FW_Algorithm}. 

\begin{proof}[Proof of Theorem \ref{thm:RSF_Competitiveness}]
    The proof of Theorem \ref{thm:RSF_Competitiveness} follows the same lines as Theorem \ref{thm:RST_Competitiveness} for Steiner tree.
    
    Having used the framework of Section \ref{sec:framework}, and observing that $i$ is a major iteration, we can apply Lemma \ref{lem:FW_BoundingFrameworkCost}:
    \begin{enumerate}
        \item The cost of the first $i$ iterations is at most $O(1)\cdot \OPT + (12\gamma + 2)\offc_{i-1}$. Using Proposition \ref{prop:RSF_OPTAndDistAtLeastPartial}, this expression is at most $O(1)\cdot \OPT + O(1)\cdot D$.
        
        \item The cost of the last iterations, from $i+1$ onwards, is at most $O(1)\cdot \max\pc{\OST_{m-1},\OST_{m}}$. Using Lemma \ref{lem:RSF_BoundingExpensiveON}, this is at most $O\pr{\log\pr{\min\pc{|\reqs|,\Delta}}} \cdot\OPT + O(1) \cdot D$.
    \end{enumerate}
    
    Overall, we have that 
    \[
        \ALG \le O\pr{\log\pr{\min\pc{|\reqs|,\Delta}}} \cdot\OPT + O(1) \cdot D.  
    \]
\end{proof}

It remains to prove Proposition \ref{prop:RSF_OPTAndDistAtLeastPartial} and Lemma \ref{lem:RSF_BoundingExpensiveON}.

\begin{prop}[Analog of Proposition \ref{prop:RST_OPTAndDistAtLeastPartial}]
	\label{prop:RSF_OPTAndDistAtLeastPartial}	
	$\offc_{i-1} \le \OPT + D $.
\end{prop}
\begin{proof}
	Identical to the proof of Proposition \ref{prop:RST_OPTAndDistAtLeastPartial} for Steiner tree with predictions.
\end{proof}


\begin{lem}[Analog of Lemma \ref{lem:RST_BoundingExpensiveON}]
    \label{lem:RSF_BoundingExpensiveON}
	It holds that 
	\[
	    \max\pc{\OSF_{m-1},\OSF_{m}}\le O\pr{\log \pr{\min\pc{\pa{\reqs}, \Delta}+2}}\cdot \OPT + O(1)\cdot D.
	\]
\end{lem}

\begin{proof}
	The proof is nearly identical to that of Lemma \ref{lem:RST_BoundingExpensiveON} for Steiner tree. The only thing to note explicitly is that for Steiner tree, the cost of serving a request is at most its distance to any previous request, while $\OSF$ can pay at most twice that distance.
\end{proof}

\subsection{Subset Competitiveness of \texorpdfstring{$\OSF$}{ON-SF}.}
\label{subsec:SF_BC_SubsetCompetitiveness}

This subsection is devoted to proving that the online algorithm $\OSF$ is subset competitive, as stated in the following lemma.

\begin{lem}
	At any point during the algorithm, $M_j$ is acyclic.
\end{lem}

\begin{proof}
	Suppose for contradiction that at some iteration in the algorithm, an edge is created between two balls $B_1,B_2 \in D_j$ that were already connected in $M_j$ at the beginning of that iteration, thus closing a cycle in $M_j$. Denote by $a_1,a_2$ the centers of $B_1,B_2$ respectively. Due to Line \ref{line:SF_BC_ExtraConnection}, we have that at the beginning of the iteration, $a_1$ and $a_2$ were connected in $F$. 
	
	Let $(s,t)$ be the terminal pair considered in that iteration, and without loss of generality assume that $B(s,2^{j-2})$ intersected $B_1$ and $B(t,2^{j-2})$ intersected $B_2$. Then $d(s,a_1) < 2^{j-1}$ and $d(t,a_2) < 2^{j-1}$. But then since $a_1$ and $a_2$ are connected in $F$, we have that $c\left(P_{s,t}^F\right) < 2^j$ at the beginning of the iteration, in contradiction to the algorithm creating an edge in $M_j$ in that iteration.
\end{proof}

Denote by $n_j$ the number of iterations in which the variable $\ell$ was set to $j$. We also denote by $n^\prime _j$ the number of such iterations on the requests of $\reqs'$.
\begin{cor}
	\label{cor:SF_BC_IterationsBoundedByDual}
	At the end of the algorithm, for every $j$, we have that $ n_j \le 2\cdot |D_j| $.
\end{cor}
\begin{proof}
	In each iteration counted in $n_j$, either a node or an edge is added to $M_j$. Since $M_j$ is acyclic, the number of its edges is at most the number of its nodes.
\end{proof}

Define $m^+:= \floor{\log \pr{\OSF(\reqs')}}$ and  $m^- := m^+ -\ceil{\log \pa{\reqs'}} - 2$.
\begin{lem}
	\label{lem:SF_BC_CostBoundedByIterations}
	$\OSF(\reqs') \le 24 \cdot \sum_{j=m^-}^{m^+} n'_j \cdot 2^{j-2}$.
\end{lem}
\begin{proof}
	First, observe that $\OSF(\reqs') \le 3 \cdot \sum_{j=-\infty}^\infty n'_j \cdot 2^j$. This is since an iteration counted in $n'_j$ costs at most $2^{j+1}$ in adding a path in Line \ref{line:SF_BC_ConnectRequest}, and at most $2^j$ in connecting conflicting nodes in Line \ref{line:SF_BC_ExtraConnection}.
	
	Now, observe that $n'_j = 0$ for every $j>m^+$. This is since an iteration counted in $n'_j$ for such a $j$ would cost strictly more than $\OSF (\reqs')$, in contradiction. Additionally, it holds that 
	\[ 
	\sum_{j=-\infty}^{m^- -1} n'_j \cdot 2^j \le |\reqs| \cdot 2^{m^- -1} \le \pa{\reqs'} \cdot \frac{\OSF(\reqs')}{8 \cdot \pa{\reqs'}} = \frac{\OSF(\reqs')}{8}. \]

	Combining those observations, we get that 
	\begin{align*}
		\OSF(\reqs') 
		\le 3 \cdot \sum_{j=-\infty}^{m^- -1} n'_j \cdot 2^j + 3 \cdot \sum_{j=m^-}^{m^+} n'_j \cdot 2^j 
		\le \frac{\OSF(\reqs')}{2}  + 3 \cdot \sum_{j=m^-}^{m^+} n'_j \cdot 2^j,
	\end{align*}
	and thus $\OSF(\reqs') \le 6 \cdot \sum_{j=m^-}^{m^+} n'_j \cdot 2^j = 24 \cdot \sum_{j=m^-}^{m^+} n'_j \cdot 2^{j-2}$.
\end{proof}

\begin{prop}
	\label{prop:SF_BC_DisjointBallsLowerBoundOptimum}
	For every $j\in \mathbb{Z}$, we have $\OPT \ge |D_j| \cdot 2^{j-2}$.
\end{prop}

\begin{proof}
	Each $D_j$ is a set of disjoint balls of radius $2^{j-2}$. Each such ball $B$ is centered at some terminal $a$ that must be connected outside the ball. Therefore, any feasible solution must contain edges (or parts of edges) in $B$ in order to connect $a$ outwards, at a total cost that is at least the radius of the ball. Since the balls are disjoint, these costs sum up in any feasible solution.
\end{proof}
\begin{proof}[Proof of Lemma \ref{lem:SF_BermanCoulstonSubsetCompetitive}]
	Using Proposition \ref{prop:SF_BC_DisjointBallsLowerBoundOptimum} for each $j\in[m^-,m^+]$ and averaging, we get that 
	\[ \OPT 
	\ge \frac{\sum_{j=m^-}^{m^+} |D_j| \cdot 2^{j-2}}{m^+ -m^- +1} 
	= \frac{\sum_{j=m^-}^{m^+} |D_j| \cdot 2^{j-2}}{\log |\reqs'| +3}. \]
	Using Lemma  \ref{lem:SF_BC_CostBoundedByIterations} and Corollary \ref{cor:SF_BC_IterationsBoundedByDual}, we have 
	\[ 
    	\OSF(\reqs') \le 24 \cdot \sum_{j=m^-}^{m^+} n_j ^\prime \cdot 2^{j-2} \le 
    	 48 \sum_{j=m^-}^{m^+} |D_j| \cdot 2^{j-2} \le O\pr{\log \pr{\pa{\reqs'}+2}}\cdot \OPT. 
	 \]
\end{proof}

\section{Online Facility Location with Predictions}
\label{sec:FL}

In the online facility location problem, we are given an undirected graph $G=(V,E)$ offline with cost $c_e$ for edge $e$ and facility opening cost $f_v$ for vertex $v$. Requests (called {\em clients}) arrive at the vertices of the graph, one per online step. The algorithm maintains a set of open facilities at a subset of the vertices of the graph. Upon the arrival of a client $\req$, the algorithm must connect the client to its closest open facility, at a connection cost equal to the shortest path in the graph between $\req$ and that facility. The algorithm is also allowed to open new facilities before making this connection; opening a new facility at a vertex $v$ adds a cost of $f_v$ to the solution. The goal is to minimize overall cost defined as the sum of opening costs of facilities and connection costs of clients. 

We would like to use the framework of Section~\ref{sec:framework}. Thus, we must describe the relevant components, i.e. a prize-collecting algorithm and a subset-competitive online algorithm.

{\bf Prize-collecting offline algorithm.}
The following lemma is due to Xu and Xu~\cite{DBLP:journals/jco/XuX09}.

\begin{lem}
    \label{fact:FL_PrizeCollecting}
	There exists a $\gamma_{\FL}$-approximation for prize-collecting facility location, where $\gamma_{\FL} = 1.8526$.
\end{lem}

{\bf Subset-competitive online algorithm.}
%
%
Algorithm \ref{alg:FL_Fotakis} (we call it $\OFL$) is an $O(\log \pr{\pa{\reqs} + 2})$-competitive deterministic algorithm for online facility location due to Fotakis~\cite{Fotakis07}. 
Here, $d(u, v)$ denotes the shortest path distance between vertices $u$ and $v$, and for a subset of vertices $F$, we denote $d(F, v) := \min_{u\in F} d(u, v)$. For any value $d$, we use $d_+ := \max(d, 0)$.

$F\gets\emptyset$

\newcommand{\UpdatePotentials}{\textsc{UpdatePotentials}\xspace}
\newcommand{\ComputeNewPotentials}{\textsc{ComputeNewPotentials}\xspace}

\begin{algorithm}[h]
	\begin{footnotesize}
		\caption{\label{alg:FL_Fotakis}Fotakis' Algorithm for Online Facility Location ($\OFL$)}

		\SetKwFunction{UpdatePotentials}{UpdatePotentials}
		\SetKwFunction{ComputeNewPotentials}{ComputeNewPotentials}

		\Input{$G=(V,E)$ -- the input graph\\
		}
		\BlankLine

		\Initialization{}{
			$F\gets \emptyset$, $L\gets \emptyset$

			\ForEach{$v\in V$}{set $p(v)\gets 0$}
		}

		\BlankLine

		\EFn(\tcp*[h]{Upon the next request $\req$ in the sequence}){\UponRequest{$r$}}{

			Set $L\gets L\cup{\req}$

			\UpdatePotentials{$F$,$\req$}

			$w \gets \arg\max_{v\in V} (p(v)-f_v)$

			\If{$ p(w) > f_w $}{
				set $ F \gets F\cup{w} $

				\ComputeNewPotentials{$F$,$L$}
			}
		}

		\BlankLine

		\Fn{\UpdatePotentials{$F$,$\req$}}{

			\ForEach{$ v\in V $}
			{
				set $ p(v)\gets p(v)+\pr{ d(F,\req) -d(v,\req) }_+$
			}

		}

		\BlankLine

		\Fn{\ComputeNewPotentials{$F$,$L$}}{
			\ForEach{$ v\in V $}
			{
				set $ p(v)\gets \sum_{r\in L}\pr{ d(F,r) -d(v,r)}_+$
			}
		}
	\end{footnotesize}
\end{algorithm}

This algorithm is \emph{not} subset-competitive with respect to its actual cost; however, it is subset-competitive with respect to an amortized cost, which is sufficient to achieve the desired results. 

We use subscript $i$ to denote the value of a variable in Algorithm \ref{alg:FL_Fotakis} immediately after handling the $i$'th request (for $i=0$ this refers to the initial value of the variable). This is used in this section for the variables $F,L$ and the potential function $p$.

\begin{defn}[Amortized cost]
	\label{defn:FL_AmortizedCost}
	For each request $\req_i \in \reqs$, define the \emph{amortized cost} of $\req_i$ as 
	\[ 
	    \amort(\req_i)= 2\cdot \min\pc{ d(F_{i-1}, \req_i), \min_{v\in V}\pr{f_v - p_{i-1}(v) + d(v,\req_i)} }.
	\]
\end{defn}

In words, the amortized cost of a request $\req$ is twice the minimum between its distance to a facility which is already open upon $\req$'s arrival, and its distance to a facility $v$ which is not open plus the ``remaining'' cost to opening that facility ($f_v$ minus the potential $p(v)$).

The following two lemmas, the proof of which appears in Subsection \ref{subsec:FL_SubsetCompetitiveness}, state that $\amort$ is indeed an amortization of the algorithm's cost, and that $\amort$ is subset competitive.
\begin{lem}
	\label{lem:FL_Fotakis_TotalCostBoundedByAmortizedCost}
	$\OFL \le \sum_{\req\in \reqs} \amort(\req)$.
\end{lem}

\begin{lem}
	\label{lem:FL_Fotakis_SubsetCompetitivenessProxy}
	Algorithm \ref{alg:FL_Fotakis} is subset-competitive, as required in Property \ref{asmp:FW_SubsetCompetitive}, with respect to its amortized cost. That is, for input $\reqs$ and for every subset of requests $\reqs'$ of $\reqs$, it holds that 
	\[ 
    	\sum_{\req\in \reqs} \amort(\req) \le O(\log \pr{|\reqs'|+2})\cdot \OPT.  
	\]
\end{lem}

    

\paragraph{Facility Location with Predictions.} We use the components $\OFL$ and $\PCFL$ to construct the algorithm for facility location with predictions. The algorithm uses $\OFL$ as a black box, as in the framework described in Algorithm \ref{alg:FW_Algorithm}. For the sake of describing the algorithm, and its analysis, we assume that the cost incurred by $\OFL$ upon receiving a request $\req$ is not the actual cost of opening facilities and connecting the request, but rather the \emph{amortized} cost $\amort(\req)$, as defined in Definition \ref{defn:FL_AmortizedCost}. In particular, this altered cost affects the counting of the online algorithm's cost, which affects whether a $\Partial$ solution is bought. As guaranteed in Lemmas \ref{lem:FL_Fotakis_TotalCostBoundedByAmortizedCost} and \ref{lem:FL_Fotakis_SubsetCompetitivenessProxy}, $\amort$ is indeed an amortized cost function, which is subset-competitive. Also observe that $\amort(\req)$ can be calculated by $\OFL$ upon the release of $\req$ (which is required for counting the online cost towards buying $\Partial$ solutions).

In this algorithm, the addition of the offline solution returned by $\Partial$ is done by adding the facilities of that solution immediately, and connecting the requests of that solution upon their future arrival. This postponing of costs does not affect the analysis. 

We would now like to analyze this algorithm for $(\Delta,D)$-errors. First, we must first define a $(\Delta,D)$-error completely by defining the matching cost of two requests. This matching cost, rather naturally, is defined to be the distance between the requests' nodes in the graph.

The main result of this section is the following theorem for facility location with predictions.
\begin{thm}
	\label{thm:RFL_Competitiveness}For every error $\err = (\Delta,D)$, Algorithm
	\ref{alg:FW_Algorithm} for facility location with predictions has the guarantee
	\[
	    \ALG\le O\pr{\log \pr{\min\pc{\pa{\reqs},\Delta}+2}}\cdot \OPT + O(1) \cdot D.
	\]
\end{thm}

When analyzing this algorithm, we regard the costs incurred by the online component $\OFL$ as the amortized costs, i.e. on request $\req$ the cost of $\OFL$ would be $\amort(\req)$. Since the sum of those costs upper-bounds the actual cost of the algorithm (Lemma \ref{lem:FL_Fotakis_TotalCostBoundedByAmortizedCost}), such an analysis is legal. Lemma \ref{lem:FL_Fotakis_SubsetCompetitivenessProxy} now implies that this amortized cost is subset-competitive, and thus the properties and lemmas of Section \ref{sec:framework} hold.

As for Steiner tree and Steiner forest, we denote by $\mreqs[\err]$ and $\mpreds[\err]$ the matched requests and the matched predictions of the error $\err$.

The proof follows the same lines as the guarantee for robust Steiner tree (Theorem~\ref{thm:RST_Competitiveness}). As in the robust Steiner tree case, define $k = |\mreqs|$. Now, fix $i$ to be the iteration in which the variable $u$ is first assigned a value which is at least $\pa{\preds} - k$ in Line~\ref{line:FW_DefiningHatK} of \cref{alg:FW_Algorithm}.

\begin{prop}[Analog of Proposition \ref{prop:RST_OPTAndDistAtLeastPartial}]
	\label{prop:RFL_OPTAtLeastPartial}	
	$\offc_{i-1} \le \OPT + D  $.
\end{prop}
\begin{proof}
    Let $i' < i$ be the iteration in which $\offc_{i-1}$ was set, i.e. the first iteration such that $\offc_{i'} = \offc_{i-1}$. 
    If $i' = 0$, then $\offc_{i-1}$ is the initial value of $\offc$, which is $0$, and the proposition holds. 
    Henceforth assume that $i' > 0$.
    
    From the definition of $i$, we have that $u_{i'} > \pa{\preds}-k$, where $u_{i'}$ is the value of the variable $u$ after iteration $i'$. 
    Thus, $c(\Partial(\preds,\pa{\preds} - k)) > 3\gamma \offc_{i'} = 3\gamma \offc_{i-1}$. From Lemma \ref{lem:FW_Partial}, this implies that the least expensive solution which satisfies at least $k$ requests from $\preds$ costs at least $\offc_{i-1}$. 
    Since $k=\pa{\mpreds}=\pa{\mreqs}$, we have that $\OPT_{\mpreds}$ is such a solution, where $\OPT_{\mpreds}$ is the optimal solution for $\mpreds$. Therefore, $c(\OPT_{\mpreds}) \ge \offc_{i-1}$.
    
    Now, observe that $\OPT$ can be augmented to a solution for $\mpreds$ through disconnecting each request $\req \in \mreqs$ from the facility $v_{\req}$ to which it is connected, and instead connecting the request $\err(\req)\in \mpreds$ to $v_{\req}$. 
    This connects all requests in $\mpreds$, and increases the cost of \OPT by exactly
    \[
        \sum_{\req\in \mreqs} d(\err(\req),v_{\req}) -  d(\req,v_{\req}) \underbrace{\le}_{\text{triangle inequality}} \sum_{\req \in \mreqs} d(\req,\err(\req)) = D
    \]
    
    We therefore have a solution for $\mpreds$ which costs at most $\OPT+D$, completing the proof.

\end{proof}

\begin{lem}[Analog of Lemma \ref{lem:RST_BoundingExpensiveON}]
    \label{lem:RFL_BoundingExpensiveON}
	It holds that 
	\[
	    \max\pc{\OFL_{m-1},\OFL_{m}}\le O\pr{\log\pr{\min\pc{\pa{\reqs},\Delta} + 2 }}\cdot\OPT + O(1) \cdot D.
	\]
\end{lem}

\begin{proof}

    Let $j\in \pc{m-1,m}$, and let $Q\subseteq \reqs$ be the subsequence of requests considered in $\OST_j$. Denote by $\preds'\subseteq \preds$ the subset of predicted requests that are satisfied by the $\Partial$ solution considered in iteration $i$ (the solution whose facilities were opened at the end of iteration $i$).
    
    The online algorithm $\OFL_j$ operates on a modified input, in which the cost of a set of facilities $F_0$ is set to $0$.
    
    We partition $Q$ into the following subsequences:
    \begin{enumerate}
        \item $Q_1 = Q \cap \mreqs$. We further partition $Q_1$ into the following sets:
        \begin{enumerate}
            \item $Q_{1,1} = \pc{\req\in Q_1 | \err(\req) \in \preds'}$
            \item $Q_{1,2} = \pc{\req\in Q_1 | \err(\req) \notin  \preds\backslash \preds'}$
        \end{enumerate}
        
        \item $Q_2 =  Q\backslash \mreqs$.
    \end{enumerate}
    
    Observe that $\OFL_j(Q) = \OFL_j(Q_{1,1}) + \OFL_j(Q_{1,2} \cup Q_2)$. We now bound each component separately.
    
    \paragraph{Bounding $\OFL_j(Q_{1,1})$.} Consider a request $\req \in Q_{1,1}$, and note from the definition of amortized cost that $\OFL_j(\req) \le 2d(\req, F')$, where $F'$ is the set of facilities which are either open or have cost $0$ immediately before considering $\req$.
    Thus, we have that 
    \[
        \OFL_j(Q_{1,1}) \le \sum_{\req \in Q_{1,1}} 2d(\req,F_0) \le \sum_{\req \in Q_{1,1}} 2d(\req,\err(\req)) + 2d(\err(\req), F_0) \le 2D + \sum_{\req \in \preds'} d(\req, F_0)
    \]
    Now, observe that $F_0$ contains all facilities bought by the $\Partial$ solution at iteration $i$, and thus $\sum_{\req \in \mpreds} d(\req, F_0)$ is at most the connection cost of that solution, which is at most the solution's cost. 
    From the proof of Item 1 of Lemma \ref{lem:FW_BoundingFrameworkCost}, We know that the cost of that solution is at most $3\gamma_{\FL} (2\offc_{i-1} + O(1)\cdot \OPT) = O(1)\cdot \OPT + O(1)\cdot D$. Overall, we have that $ \OFL_j(Q_{1,1}) \le O(1)\cdot \OPT + O(1)\cdot D$.
    
    \paragraph{Bounding $\OFL_j(Q_{1,2} \cup Q_2)$.} Using Property~\ref{asmp:FW_SubsetCompetitive} (subset competitiveness), we have that 
    
	\[ 
	    \OFL_j(Q_{1,2} \cup Q_2) \le O\pr{\log \pr{\pa{Q_{1,2}\cup Q_2}+2}} \cdot \OPT'  
	\]
	where $\OPT'$ is the optimal solution to $Q$ with the cost of $F_0$ set to $0$. Clearly, $\OPT' \le \OPT$. 
	
    Now, observe that $\pa{Q_{1,2}} \le \pa{\preds\backslash \preds'}$. From the definition of the major iteration $i$, and from Lemma \ref{lem:FW_Partial}, it holds that $\pa{\preds\backslash \preds'} \le 2\gamma_{\FL} \cdot (|\preds| - k) = 2\gamma_{\FL} \pa{\preds\backslash \mpreds}$. As for $Q_2$, it holds that $Q_2 \le \pa{\reqs \backslash \mreqs}$. These facts imply that 
    \[ 
        \pa{Q_{1,2}\cup Q_2} \le 2\gamma_{\FL} \Delta
    \]
    In addition, it clearly holds that $\pa{Q_{1,2}\cup Q_2} \le \pa{\reqs}$. Therefore, it holds that 
    \[  
        \OFL_j(Q_{1,2}\cup Q_2) \le O\pr{\log \pr{\min\pc{\pa{\reqs},\Delta}+2}}\cdot \OPT
    \]
    
    Combining this with the previous bound for $\OFL_j(Q_{1,1})$, we obtain
    \[ 
        \OFL_j(Q) \le O\pr{\log \pr{\min\pc{\pa{\reqs},\Delta}+2}}\cdot \OPT + O(1)\cdot D
    \] 
    completing the proof.

\end{proof}

Theorem \ref{thm:RFL_Competitiveness} now follows from Lemma \ref{lem:FW_BoundingFrameworkCost}, Proposition \ref{prop:RST_OPTAndDistAtLeastPartial} and Lemma \ref{lem:RFL_BoundingExpensiveON}, using an identical argument as for Theorem \ref{thm:RST_Competitiveness}.

\subsection{Subset Competitiveness of Algorithm \ref{alg:FL_Fotakis}.}
\label{subsec:FL_SubsetCompetitiveness}

The goal of this subsection is to prove Lemmas \ref{lem:FL_Fotakis_TotalCostBoundedByAmortizedCost} and \ref{lem:FL_Fotakis_SubsetCompetitivenessProxy}. 
Henceforth, fix the input sequence $\reqs = (\req_1,\dots,\req_{|\reqs|})$, and recall that for every variable $x$ (e.g., $F$, $L$, $p$) in the algorithm, we denote by $x_i$ the value of $x$ immediately after iteration $i$ (if $i=0$ this denotes the initial value of the variable).

\begin{lem}
	\label{lem:FL_Fotakis_StablePotentials}
	For every $i\in \{1, 2, \ldots, |\reqs|\}$, and for every $v\in V $, it holds that $p_i (v) \le f_v$. 
\end{lem}

\begin{proof}
	Identical to the proof of \cite[Lemma 1]{Fotakis07}.
\end{proof}

Denote by $F^\ast$ the set of facilities bought by the optimal solution. For each facility $v^\ast \in  F^\ast$, let $C_{v^\ast}$ be the set of requests connected to $v^\ast$ by the optimal solution. 

We denote the cost incurred by the algorithm in opening facilities by $\OFL_F$, and the cost of connecting requests by $\OFL_C$. We define $\OPT_F$ and $\OPT_C$ similarly. 
For every subsequence of requests $\reqs'$, we define $\OFL_F (\reqs')$ and $\OFL_C (\reqs')$ as the total opening cost and the total connection cost incurred for the requests of $\reqs'$, respectively. We also use the same terminology for $\OPT$.

\begin{cor}
	\label{cor:FL_Fotakis_OptimalClusterDistanceBound}
	$  \pa{L_i\cap C_{v^\ast}}  \cdot d(F_i,v^\ast) \le f_{v^\ast} + 2 \cdot \OPT_C (C_{v^\ast}) $ for every $v^\ast \in F^\ast$.
\end{cor}
\begin{proof}
	Identical to the proof of \cite[Corollary 1]{Fotakis07}.
\end{proof}

\begin{prop}
	\label{prop:FL_Fotakis_PotentialFunctionCases}
	For each request $\req_i \in \reqs$, if $\OFL$ opens a facility at any $w\in V$ in iteration $i$, then $\amort(\req_i)=2 \left( f_w-p_{i-1}(w) + d(w,\req_i) \right)$. Otherwise, $\amort(\req_i) = 2 d(F_{i-1},\req_i)$.
\end{prop}
\begin{proof}
	Identical to the proof of \cite[Lemma 3]{Fotakis07}.
\end{proof}

\begin{proof}[Proof of Lemma \ref{lem:FL_Fotakis_TotalCostBoundedByAmortizedCost}]
	Define the potential function $\phi(i) = \sum_{\req\in L_i} d(F_i,\req)$ for every iteration $i$. We show by induction on $i$ that 
	\begin{equation}
		\label{eq:FL_Fotakis_InductionClaim}
		\OFL((\req_1,\dots,\req_i)) + \phi(i)\le \sum_{\req\in L_i} \amort(\req),
	\end{equation}
	thereby proving the lemma.
	Suppose Equation \ref{eq:FL_Fotakis_InductionClaim} holds for iteration $i-1$. Iteration $i$ would increase the RHS by $\amort(\req_i)$, and increase the LHS by $\OFL({\req_i}) + \Delta \phi$, where we define $\Delta \phi = \phi(i) - \phi(i-1)$ to be the change in the potential function $\phi$.
	
	If the algorithm did not open a new facility in iteration $i$, then we have $\OFL({\req_i}) = d(r_i,F_{i-1})$. In addition, we have that $\Delta \phi = d(\req_i,F_{i-1})$. In the RHS, using Proposition \ref{prop:FL_Fotakis_PotentialFunctionCases}, we have that $\amort(\req_i) = 2d(\req_i,F_{i-1})$, completing the proof for this case.
	
	For the other case, in which the algorithm opens a new facility at $w\in V$ during iteration $i$, we have that the algorithm's cost for that iteration is $f_w$ for opening the facility, plus the cost of connecting $\req_i$. 
	
	Observe that if we open a facility for the current client, it becomes the closest open facility to the client. This is since at the beginning of the iteration, we had that $p(w) \leq f_w$, using Lemma \ref{lem:FL_Fotakis_StablePotentials}. Thus, if $w$ was opened in iteration $i$, its potential has increased upon the arrival of $\req_i$. But this only happens if $w$ is closer to $\req_i$ than any open facility in $F_{i-1}$. Therefore, the connection cost of $\req_i$ in iteration $i$ is exactly $d(\req_i,w)$.
	
	As for the change in the potential function $\phi$, we consider first the addition of $w$ to $F$ and then the addition of $\req_i$ to the request set. The addition of $w$ to $F$ has reduced $\phi$ by exactly $p_{i-1}(w)$. Then, the addition of $\req_i$ increased $\phi$ by exactly $d(r_i,w)$. Thus, The LHS of Equation \ref{eq:FL_Fotakis_InductionClaim} increased by a total of $f_w - p_{i-1}(w) + d(\req_i,w)$. But according to Proposition \ref{prop:FL_Fotakis_PotentialFunctionCases}, the RHS has increased by twice that amount, and thus the equation holds.
	
	Since the potential function $\phi$ is initially $0$, and is always non-negative, the proof is complete.
\end{proof}

\begin{proof}[Proof of Lemma \ref{lem:FL_Fotakis_SubsetCompetitivenessProxy}]
	Fix a facility $f^\ast \in F^\ast$. Denote the requests of $C_{f^\ast}$ by $q_1, \dots, q_{|C_{f^\ast}|}$ ordered by their arrival. With $k' = |\reqs' \cap C_f^\ast|$, denote by $i_1,\dots,i_{k'}$ the indices in the requests of $\reqs' \cap C_f^\ast$.
	
	For $q_{i_1}$, we have that $\amort (q_{i_1}) \le 2f_{v^\ast} + 2d(v^\ast, q_{i_1})$. 
	Now, consider any $j>1$, and let $i$ be the iteration in which $q_{i_j}$ is considered.
	For $j>1$, using Corollary 	\ref{cor:FL_Fotakis_OptimalClusterDistanceBound}, we have that 
	\begin{align*}
		\amort(q_{i_j}) 
		    & \le 2\cdot d(q_{i_j},F_{i-1})\\
			& \le 2\cdot d(v^\ast,F_{i-1}) + 2\cdot d(v^\ast , q_{i_j}) \quad \text{(by triangle inequality)} \\
			& = 2\cdot d(v^\ast,F_{i-1}) + 2\cdot d(v^\ast , q_{i_j})\\
			& \le \frac{2}{|L_{i-1} \cap C_{v^\ast}|} \cdot \left( f_{v^\ast} + 2\cdot \OPT_C (C_{v^\ast}) \right) + 2\cdot d(v^\ast , q_{i_j}) \quad \text{(by Corollary \ref{cor:FL_Fotakis_OptimalClusterDistanceBound})}\\
			& \le \frac{2}{|L_{i-1} \cap C_{v^\ast} \cap \reqs'|} \cdot \left( f_{v^\ast} + 2\cdot \OPT_C (C_{v^\ast}) \right) + 2\cdot d(v^\ast , q_{i_j}) \\
			& = \frac{2}{j-1} \cdot \left( f_{v^\ast} + 2 \cdot \OPT_C (C_{v^\ast}) \right) + 2\cdot d(v^\ast , q_{i_j}) \quad \text{(by the definition of $j$)}.
	\end{align*}
%
	Summing over $j$, we get
	\begin{align*}
		\sum_{j=1}^{k'} \amort(q_{i_j}) 
		&\le  2\cdot f_{v^\ast} + \sum_{j=2}^{k'} \frac{2}{j-1} \cdot \left( f_{v^\ast} + 2\cdot  \OPT_C (C_{v^\ast}) \right)  + 2\cdot \OPT_C(C_{v^\ast} \cap R^\prime)\\
		& \le 2(\log k' + 1) \cdot f_{v^\ast} + 4(\log k' +1) \cdot \OPT_C (C_{v^\ast} \cap \reqs')\\
		& \le 2(\log k + 1) \cdot f_{v^\ast} + 4(\log k +1) \cdot \OPT_C (C_{v^\ast} \cap \reqs').
	\end{align*}

	Summing over all $v^\ast \in F^\ast$, we get that 
	\[ 
	    \sum_{\req\in \reqs'} \amort(\req) \le O(\log \pr{k+2}) \cdot \OPT. 
	\]
%
\end{proof}

\section{Online Capacitated Facility Location with Predictions}
\label{sec:SCFL}
\newcommand{\cpt}[1]{\beta_{#1}}

The results of Section \ref{sec:FL} also extend to the soft-capacitated facility location problem. 
In this problem, each node $v\in V$ has, in addition to the facility opening cost $f_v$, a \emph{capacity} $\cpt{v}$ which is a natural number.
After opening a facility at $v$, a solution can connect at most $\cpt{v}$ clients to that facility. 
We consider the soft-capacitated case, in which a solution may open multiple facilities at a single node $v$, each at a cost of $f_v$, such that each facility can connect $\cpt{v}$ requests.

The matching cost between two requests in this case is identical to the uncapacitated case -- specifically, it is the distance between the two requests in the graph.

This problem can be related to the uncapacitated facility location problem using the following folklore reduction, which we nevertheless describe for completeness: given an instance for soft-capacitated facility location, which contains a graph $G = (V,E)$, facility costs $\pc{f_v}$ and capacities $\pc{\cpt{v}}$, construct a uncapacitated facility location instance over the graph $G' = (V\cup V', E\cup E')$, where:
\begin{itemize}
    \item $V'$ contains a copy $v'$ for every $v\in V$. The opening cost in $v'$ would be $f_v$, and the opening cost in $v$ would be $\infty$.
    \item $E'$ contains an edge $e= (v,v')$ for every $v\in V$ (and its copy $v' \in V'$), such that the cost of the edge is $\frac{f_v}{\cpt{v}}$.
    \item The cost of the edges of $E$ remains as in the original instance.
\end{itemize}
Whenever a request is released in the original capacitated instance on a node $v$, release a request on $v$ in the uncapacitated instance.

In essence, this reduction restricts opening facilities to copies of the original nodes, such that these copies are somewhat distant from the rest of the graph, as to discourage connecting requests frivolously (wasting the capacity).

\newcommand{\inst}{\mathcal{I}}

The algorithm for an instance $\inst$ of capacitated facility location would therefore be:
\begin{enumerate}
    \item Reduce $\inst$ to an uncapacitated instance $\inst'$, on which the algorithm of Section \ref{sec:FL} is run.
    \item Whenever a request $\req$ is released on a node $v$ in $\inst$, release a request on $v$ in $\inst'$.
    \item Whenever the algorithm opens a facility at $v'$ in $\inst'$, open a facility at $v$ in $\inst$. 
    \item Whenever the algorithm connects a request at $\req$ to a facility at $v'$ in $\inst'$, connect $\req$ to a facility at $v$ in $\inst$ (opening an additional facility at $v$ if required). 
\end{enumerate}

Denote by $\ALG$ the cost of the algorithm on $\inst$, $\ALG'$ the cost of the algorithm of Section \ref{sec:FL} on $\inst'$, $\OPT'$ the optimal solution for $\inst'$, and $\OPT$ the optimal solution for $\inst$.
Observe that:
\begin{enumerate}
    \item $\ALG \le \ALG'$. This is since for each connection to facility $v$, $\ALG'$ pays exactly $x_v = \frac{f_v}{\cpt{v}}$ more than $\ALG$. As for the opening costs, the opening of a first facility at $v$ in $\ALG$ is charged the opening of the facility at $v'$ in $\ALG'$, and the opening of any subsequent facility at $v$ in $\ALG$ is charged to $\cpt{v}$ connections which cost $x_v$ more, for a total of $\cpt{v}\cdot x_v = f_v$.
    
    \item $\OPT' \le 2\OPT$. We show that a $\OPT$ induces a solution for $\inst'$ of at most double cost. The solution $\OPT''$ for $\inst'$ consists of opening a facility at $v'$ for every $v$ in which $\OPT$ opened a nonzero number of facilities, and connecting to $v'$ all requests which $\OPT$ connected to facilities at $v$. Let $y_v$ be the number of requests connected to node $v$ in $\OPT$. Then the connection cost of $\OPT''$ is exactly $C_v + y_v \cdot \frac{f_v}{\cpt{v}}$, where $C_v$ is the cost of connecting requests to $v$ incurred by $\OPT$. Now, observe that $y_v \cdot \frac{f_v}{\cpt{v}}$ is a lower bound on the buying cost of $\OPT$ on facilities in $v$; thus, the connection cost of $\OPT''$ is at most $\OPT$. Since the buying cost of $\OPT''$ is also clearly at most $\OPT$, we arrive at the desired conclusion.
\end{enumerate}
Importantly, observe that the distance between requests in the metric spaces of $G$ and $G'$ is identical - thus, the matching cost $D$ is the same in both $\inst$ and $\inst'$. We can thus apply Theorem \ref{thm:RFL_Competitiveness} to $\inst'$, which combined with the above conclusions, yields the following theorem.

\begin{thm}
	\label{thm:RFL_SoftCapCompetitiveness}For error $\err = (\Delta,D)$ in soft-capacitated facility location, Algorithm
	\ref{alg:FW_Algorithm} (combined with the above reduction) has the guarantee
	\[
	    \ALG\le O\pr{\log \pr{\min\pc{\pa{\reqs},\Delta}+2}}\cdot \OPT + O(1) \cdot D.
	\]
\end{thm}

\section{Conclusions}
\label{sec:conclusions}

In this paper, we presented algorithms for classical graph problems in the online problems with predictions framework. Our main contributions were: (a) a novel definition of prediction error for graph and metric problems that incorporates both the numerical value of errors and their magnitude in the metric space, (b) a new ``black box'' framework for converting online algorithms that satisfy a subset competitiveness condition and offline prize-collecting algorithms to new online algorithms with prediction, and (c) an application of this framework to a range of graph problems to obtain tight interpolations between offline and online approximations. In particular, this improves the dependence of the competitive ratio from being on the size of the instance to the number of prediction errors. We hope that the concepts introduced in this paper--the notion of metric error and the general framework for online algorithms with prediction--will be useful for other problems in the future.




\bibliographystyle{plain}
\bibliography{bibfile,dp-refs}

\begin{thebibliography}{10}

\bibitem{AgrawalKR95}
Ajit Agrawal, Philip~N. Klein, and R.~Ravi.
\newblock When trees collide: An approximation algorithm for the generalized
  steiner problem on networks.
\newblock {\em {SIAM} J. Comput.}, 24(3):440--456, 1995.

\bibitem{AnandGP20}
Keerti Anand, Rong Ge, and Debmalya Panigrahi.
\newblock Customizing ml predictions for online algorithms.
\newblock In {\em Proceedings of the 37th International Conference on Machine
  Learning, {ICML} 2020, 13-18 July 2020}, Proceedings of Machine Learning
  Research, 2020.

\bibitem{AntoniadisCEPS20}
Antonios Antoniadis, Christian Coester, Marek Elias, Adam Polak, and Bertrand
  Simon.
\newblock Online metric algorithms with untrusted predictions.
\newblock In {\em Proceedings of the 37th International Conference on Machine
  Learning, {ICML} 2020, 13-18 July 2020}, Proceedings of Machine Learning
  Research, 2020.

\bibitem{AroraK06}
Sanjeev Arora and George Karakostas.
\newblock A 2 + \emph{epsilon} approximation algorithm for the \emph{k}-mst
  problem.
\newblock {\em Math. Program.}, 107(3):491--504, 2006.

\bibitem{AryaR98}
Sunil Arya and H.~Ramesh.
\newblock A 2.5-factor approximation algorithm for the \emph{k}-mst problem.
\newblock {\em Inf. Process. Lett.}, 65(3):117--118, 1998.

\bibitem{AwerbuchABV98}
Baruch Awerbuch, Yossi Azar, Avrim Blum, and Santosh~S. Vempala.
\newblock New approximation guarantees for minimum-weight k-trees and
  prize-collecting salesmen.
\newblock {\em {SIAM} J. Comput.}, 28(1):254--262, 1998.

\bibitem{DBLP:conf/nips/BamasMS20}
{\'{E}}tienne Bamas, Andreas Maggiori, and Ola Svensson.
\newblock The primal-dual method for learning augmented algorithms.
\newblock In Hugo Larochelle, Marc'Aurelio Ranzato, Raia Hadsell,
  Maria{-}Florina Balcan, and Hsuan{-}Tien Lin, editors, {\em Advances in
  Neural Information Processing Systems 33: Annual Conference on Neural
  Information Processing Systems 2020, NeurIPS 2020, December 6-12, 2020,
  virtual}, 2020.

\bibitem{BermanC97}
Piotr Berman and Chris Coulston.
\newblock On-line algorithms for steiner tree problems (extended abstract).
\newblock In {\em Proceedings of the Twenty-Ninth Annual {ACM} Symposium on the
  Theory of Computing, El Paso, Texas, USA, May 4-6, 1997}, pages 344--353,
  1997.

\bibitem{DBLP:conf/icml/BhaskaraC0P20}
Aditya Bhaskara, Ashok Cutkosky, Ravi Kumar, and Manish Purohit.
\newblock Online learning with imperfect hints.
\newblock In {\em Proceedings of the 37th International Conference on Machine
  Learning, {ICML} 2020, 13-18 July 2020, Virtual Event}, volume 119 of {\em
  Proceedings of Machine Learning Research}, pages 822--831. {PMLR}, 2020.

\bibitem{BlumRV99}
Avrim Blum, R.~Ravi, and Santosh Vempala.
\newblock A constant-factor approximation algorithm for the \emph{k}-mst
  problem.
\newblock {\em J. Comput. Syst. Sci.}, 58(1):101--108, 1999.

\bibitem{ByrkaA10}
Jaroslaw Byrka and Karen Aardal.
\newblock An optimal bifactor approximation algorithm for the metric
  uncapacitated facility location problem.
\newblock {\em {SIAM} J. Comput.}, 39(6):2212--2231, 2010.

\bibitem{ByrkaGRS13}
Jaroslaw Byrka, Fabrizio Grandoni, Thomas Rothvo{\ss}, and Laura Sanit{\`{a}}.
\newblock Steiner tree approximation via iterative randomized rounding.
\newblock {\em J. {ACM}}, 60(1):6:1--6:33, 2013.

\bibitem{CharikarG05}
Moses Charikar and Sudipto Guha.
\newblock Improved combinatorial algorithms for facility location problems.
\newblock {\em {SIAM} J. Comput.}, 34(4):803--824, 2005.

\bibitem{CharikarNS04}
Moses Charikar, Joseph Naor, and Baruch Schieber.
\newblock Resource optimization in qos multicast routing of real-time
  multimedia.
\newblock {\em {IEEE/ACM} Trans. Netw.}, 12(2):340--348, 2004.

\bibitem{ChudakS03}
Fabi{\'{a}}n~A. Chudak and David~B. Shmoys.
\newblock Improved approximation algorithms for the uncapacitated facility
  location problem.
\newblock {\em {SIAM} J. Comput.}, 33(1):1--25, 2003.

\bibitem{ChuzhoyGNS08}
Julia Chuzhoy, Anupam Gupta, Joseph Naor, and Amitabh Sinha.
\newblock On the approximability of some network design problems.
\newblock {\em {ACM} Trans. Algorithms}, 4(2):23:1--23:17, 2008.

\bibitem{DBLP:conf/nips/DekelFHJ17}
Ofer Dekel, Arthur Flajolet, Nika Haghtalab, and Patrick Jaillet.
\newblock Online learning with a hint.
\newblock In Isabelle Guyon, Ulrike von Luxburg, Samy Bengio, Hanna~M. Wallach,
  Rob Fergus, S.~V.~N. Vishwanathan, and Roman Garnett, editors, {\em Advances
  in Neural Information Processing Systems 30: Annual Conference on Neural
  Information Processing Systems 2017, December 4-9, 2017, Long Beach, CA,
  {USA}}, pages 5299--5308, 2017.

\bibitem{Fotakis07}
Dimitris Fotakis.
\newblock A primal-dual algorithm for online non-uniform facility location.
\newblock {\em J. Discrete Algorithms}, 5(1):141--148, 2007.

\bibitem{Fotakis08}
Dimitris Fotakis.
\newblock On the competitive ratio for online facility location.
\newblock {\em Algorithmica}, 50(1):1--57, 2008.

\bibitem{Garg96}
Naveen Garg.
\newblock A 3-approximation for the minimum tree spanning k vertices.
\newblock In {\em 37th Annual Symposium on Foundations of Computer Science,
  {FOCS} '96, Burlington, Vermont, USA, 14-16 October, 1996}, pages 302--309,
  1996.

\bibitem{Garg05}
Naveen Garg.
\newblock Saving an epsilon: a 2-approximation for the k-mst problem in graphs.
\newblock In {\em Proceedings of the 37th Annual {ACM} Symposium on Theory of
  Computing, Baltimore, MD, USA, May 22-24, 2005}, pages 396--402, 2005.

\bibitem{GoemansW95}
Michel~X. Goemans and David~P. Williamson.
\newblock A general approximation technique for constrained forest problems.
\newblock {\em {SIAM} J. Comput.}, 24(2):296--317, 1995.

\bibitem{GollapudiP19}
Sreenivas Gollapudi and Debmalya Panigrahi.
\newblock Online algorithms for rent-or-buy with expert advice.
\newblock In {\em Proceedings of the 36th International Conference on Machine
  Learning, {ICML} 2019, 9-15 June 2019, Long Beach, California, {USA}}, pages
  2319--2327, 2019.

\bibitem{GuhaK99}
Sudipto Guha and Samir Khuller.
\newblock Greedy strikes back: Improved facility location algorithms.
\newblock {\em J. Algorithms}, 31(1):228--248, 1999.

\bibitem{GuptaHNR10}
Anupam Gupta, Mohammad~Taghi Hajiaghayi, Viswanath Nagarajan, and R.~Ravi.
\newblock Dial a ride from \emph{k}-forest.
\newblock {\em {ACM} Trans. Algorithms}, 6(2):41:1--41:21, 2010.

\bibitem{HajiaghayiJ06}
Mohammad~Taghi Hajiaghayi and Kamal Jain.
\newblock The prize-collecting generalized steiner tree problem via a new
  approach of primal-dual schema.
\newblock In {\em Proceedings of the Seventeenth Annual {ACM-SIAM} Symposium on
  Discrete Algorithms, {SODA} 2006, Miami, Florida, USA, January 22-26, 2006},
  pages 631--640, 2006.

\bibitem{Hochbaum82}
Dorit~S. Hochbaum.
\newblock Heuristics for the fixed cost median problem.
\newblock {\em Math. Program.}, 22(1):148--162, 1982.

\bibitem{hsu2018learningbased}
Chen-Yu Hsu, Piotr Indyk, Dina Katabi, and Ali Vakilian.
\newblock Learning-based frequency estimation algorithms.
\newblock In {\em International Conference on Learning Representations}, 2019.

\bibitem{DBLP:conf/spaa/Im0QP21}
Sungjin Im, Ravi Kumar, Mahshid~Montazer Qaem, and Manish Purohit.
\newblock Non-clairvoyant scheduling with predictions.
\newblock In Kunal Agrawal and Yossi Azar, editors, {\em {SPAA} '21: 33rd {ACM}
  Symposium on Parallelism in Algorithms and Architectures, Virtual Event, USA,
  6-8 July, 2021}, pages 285--294. {ACM}, 2021.

\bibitem{ImaseW91}
Makoto Imase and Bernard~M. Waxman.
\newblock Dynamic steiner tree problem.
\newblock {\em {SIAM} J. Discrete Math.}, 4(3):369--384, 1991.

\bibitem{IndykVY19}
Piotr Indyk, Ali Vakilian, and Yang Yuan.
\newblock Learning-based low-rank approximations.
\newblock {\em CoRR}, abs/1910.13984, 2019.

\bibitem{JainMMSV03}
Kamal Jain, Mohammad Mahdian, Evangelos Markakis, Amin Saberi, and Vijay~V.
  Vazirani.
\newblock Greedy facility location algorithms analyzed using dual fitting with
  factor-revealing {LP}.
\newblock {\em J. {ACM}}, 50(6):795--824, 2003.

\bibitem{JainMS02}
Kamal Jain, Mohammad Mahdian, and Amin Saberi.
\newblock A new greedy approach for facility location problems.
\newblock In {\em Proceedings on 34th Annual {ACM} Symposium on Theory of
  Computing, May 19-21, 2002, Montr{\'{e}}al, Qu{\'{e}}bec, Canada}, pages
  731--740, 2002.

\bibitem{JainV01}
Kamal Jain and Vijay~V. Vazirani.
\newblock Approximation algorithms for metric facility location and
  \emph{k}-median problems using the primal-dual schema and lagrangian
  relaxation.
\newblock {\em J. {ACM}}, 48(2):274--296, 2001.

\bibitem{JiangPS20}
Zhihao Jiang, Debmalya Panigrahi, and Kevin Sun.
\newblock Online algorithms for weighted paging with predictions.
\newblock In Artur Czumaj, Anuj Dawar, and Emanuela Merelli, editors, {\em 47th
  International Colloquium on Automata, Languages, and Programming, {ICALP}
  2020, July 8-11, 2020, Saarbr{\"{u}}cken, Germany (Virtual Conference)},
  volume 168 of {\em LIPIcs}, pages 69:1--69:18. Schloss Dagstuhl -
  Leibniz-Zentrum f{\"{u}}r Informatik, 2020.

\bibitem{KarpinskiZ97}
Marek Karpinski and Alexander Zelikovsky.
\newblock New approximation algorithms for the steiner tree problems.
\newblock {\em J. Comb. Optim.}, 1(1):47--65, 1997.

\bibitem{KorupoluPR00}
Madhukar~R. Korupolu, C.~Greg Plaxton, and Rajmohan Rajaraman.
\newblock Analysis of a local search heuristic for facility location problems.
\newblock {\em J. Algorithms}, 37(1):146--188, 2000.

\bibitem{LattanziLMV20}
Silvio Lattanzi, Thomas Lavastida, Benjamin Moseley, and Sergei Vassilvitskii.
\newblock Online scheduling via learned weights.
\newblock In {\em Proceedings of the Thirty-First Annual {ACM-SIAM} Symposium
  on Discrete Algorithms, {SODA} 2020, New Orleans, LA, USA, January 5 - 8,
  2020.}, 2020.

\bibitem{Li13}
Shi Li.
\newblock A 1.488 approximation algorithm for the uncapacitated facility
  location problem.
\newblock {\em Inf. Comput.}, 222:45--58, 2013.

\bibitem{lykouris2018competitive}
Thodoris Lykouris and Sergei Vassilvtiskii.
\newblock Competitive caching with machine learned advice.
\newblock In {\em International Conference on Machine Learning}, pages
  3302--3311, 2018.

\bibitem{DBLP:conf/random/MahdianYZ03}
Mohammad Mahdian, Yingyu Ye, and Jiawei Zhang.
\newblock A 2-approximation algorithm for the soft-capacitated facility
  location problem.
\newblock In {\em Approximation, Randomization, and Combinatorial Optimization:
  Algorithms and Techniques, 6th International Workshop on Approximation
  Algorithms for Combinatorial Optimization Problems, {APPROX} 2003 and 7th
  International Workshop on Randomization and Approximation Techniques in
  Computer Science, {RANDOM} 2003, Princeton, NJ, USA, August 24-26, 2003,
  Proceedings}, volume 2764 of {\em Lecture Notes in Computer Science}, pages
  129--140. Springer, 2003.

\bibitem{MahdianYZ06}
Mohammad Mahdian, Yinyu Ye, and Jiawei Zhang.
\newblock Approximation algorithms for metric facility location problems.
\newblock {\em {SIAM} J. Comput.}, 36(2):411--432, 2006.

\bibitem{MedinaV17}
Andres~Mu{\~{n}}oz Medina and Sergei Vassilvitskii.
\newblock Revenue optimization with approximate bid predictions.
\newblock In Isabelle Guyon, Ulrike von Luxburg, Samy Bengio, Hanna~M. Wallach,
  Rob Fergus, S.~V.~N. Vishwanathan, and Roman Garnett, editors, {\em Advances
  in Neural Information Processing Systems 30: Annual Conference on Neural
  Information Processing Systems 2017, 4-9 December 2017, Long Beach, CA,
  {USA}}, pages 1858--1866, 2017.

\bibitem{Meyerson01}
Adam Meyerson.
\newblock Online facility location.
\newblock In {\em FOCS}, pages 426--431, 2001.

\bibitem{MirrokniGZ12}
Vahab~S. Mirrokni, Shayan~Oveis Gharan, and Morteza Zadimoghaddam.
\newblock Simultaneous approximations for adversarial and stochastic online
  budgeted allocation.
\newblock In {\em Proceedings of the Twenty-Third Annual {ACM-SIAM} Symposium
  on Discrete Algorithms, {SODA} 2012, Kyoto, Japan, January 17-19, 2012},
  pages 1690--1701, 2012.

\bibitem{mitzenmacher2018model}
Michael Mitzenmacher.
\newblock A model for learned bloom filters and optimizing by sandwiching.
\newblock In {\em Advances in Neural Information Processing Systems}, pages
  464--473, 2018.

\bibitem{Mitzenmacher20}
Michael Mitzenmacher.
\newblock Scheduling with predictions and the price of misprediction.
\newblock In Thomas Vidick, editor, {\em 11th Innovations in Theoretical
  Computer Science Conference, {ITCS} 2020, January 12-14, 2020, Seattle,
  Washington, {USA}}, volume 151 of {\em LIPIcs}, pages 14:1--14:18. Schloss
  Dagstuhl - Leibniz-Zentrum f{\"{u}}r Informatik, 2020.

\bibitem{PromelS00}
Hans~J{\"{u}}rgen Pr{\"{o}}mel and Angelika Steger.
\newblock A new approximation algorithm for the steiner tree problem with
  performance ratio 5/3.
\newblock {\em J. Algorithms}, 36(1):89--101, 2000.

\bibitem{purohit2018improving}
Manish Purohit, Zoya Svitkina, and Ravi Kumar.
\newblock Improving online algorithms via ml predictions.
\newblock In {\em Advances in Neural Information Processing Systems}, pages
  9661--9670, 2018.

\bibitem{RobinsZ05}
Gabriel Robins and Alexander Zelikovsky.
\newblock Tighter bounds for graph steiner tree approximation.
\newblock {\em {SIAM} J. Discrete Math.}, 19(1):122--134, 2005.

\bibitem{DBLP:conf/soda/Rohatgi20}
Dhruv Rohatgi.
\newblock Near-optimal bounds for online caching with machine learned advice.
\newblock In Shuchi Chawla, editor, {\em Proceedings of the 2020 {ACM-SIAM}
  Symposium on Discrete Algorithms, {SODA} 2020, Salt Lake City, UT, USA,
  January 5-8, 2020}, pages 1834--1845. {SIAM}, 2020.

\bibitem{SegevS06}
Danny Segev and Gil Segev.
\newblock Approximate \emph{k}-steiner forests via the lagrangian relaxation
  technique with internal preprocessing.
\newblock In {\em Algorithms - {ESA} 2006, 14th Annual European Symposium,
  Zurich, Switzerland, September 11-13, 2006, Proceedings}, pages 600--611,
  2006.

\bibitem{ShmoysTA97}
David~B. Shmoys, {\'{E}}va Tardos, and Karen Aardal.
\newblock Approximation algorithms for facility location problems (extended
  abstract).
\newblock In {\em Proceedings of the Twenty-Ninth Annual {ACM} Symposium on the
  Theory of Computing, El Paso, Texas, USA, May 4-6, 1997}, pages 265--274,
  1997.

\bibitem{Sviridenko02}
Maxim Sviridenko.
\newblock An improved approximation algorithm for the metric uncapacitated
  facility location problem.
\newblock In {\em Integer Programming and Combinatorial Optimization, 9th
  International {IPCO} Conference, Cambridge, MA, USA, May 27-29, 2002,
  Proceedings}, pages 240--257, 2002.

\bibitem{Umboh2015}
Seeun Umboh.
\newblock Online network design algorithms via hierarchical decompositions.
\newblock In {\em Proceedings of the Twenty-sixth Annual ACM-SIAM Symposium on
  Discrete Algorithms}, SODA '15, pages 1373--1387, Philadelphia, PA, USA,
  2015. Society for Industrial and Applied Mathematics.

\bibitem{DBLP:conf/approx/Wei20}
Alexander Wei.
\newblock Better and simpler learning-augmented online caching.
\newblock In Jaroslaw Byrka and Raghu Meka, editors, {\em Approximation,
  Randomization, and Combinatorial Optimization. Algorithms and Techniques,
  {APPROX/RANDOM} 2020, August 17-19, 2020, Virtual Conference}, volume 176 of
  {\em LIPIcs}, pages 60:1--60:17. Schloss Dagstuhl - Leibniz-Zentrum f{\"{u}}r
  Informatik, 2020.

\bibitem{DBLP:journals/jco/XuX09}
Guang Xu and Jinhui Xu.
\newblock An improved approximation algorithm for uncapacitated facility
  location problem with penalties.
\newblock {\em J. Comb. Optim.}, 17(4):424--436, 2009.

\bibitem{Zelikovsky93}
Alexander Zelikovsky.
\newblock An 11/6-approximation algorithm for the network steiner problem.
\newblock {\em Algorithmica}, 9(5):463--470, 1993.

\end{thebibliography}

\appendix

\section{Additional Proofs from Section \ref{sec:framework}}
\label{sec:framework_AdditionalProofs}

\begin{proof}[Proof of Lemma \ref{lem:FW_Partial}]
    If $\gamma u \ge |\preds|$, then $\Partial$ returns $\emptyset$, which satisfies both claims of the lemma. Assume henceforth that $\gamma u < |\preds|$. 
    
    To prove the first claim, consider the two cases in the subroutine. If the returned solution is $S_2$, then the number of unsatisfied requests is $u_2$, which is at most $\gamma u \le 2\gamma u$, as required. Otherwise, $S_1$ is returned, and thus $\gamma u \ge \frac{u_1+u_2}{2}$. This implies that $2\gamma u \ge u_1+u_2 \ge u_1$, which completes the proof of the first claim.
    	
    We now prove the second claim. Fix $i$ to be the index chosen in partial, such that $S_2 = \PC(\preds,\pi_{2^i})$ does not satisfy at most $\gamma u$ requests from $\preds$. 
    Fix $u^*$ to be the number of requests in $\preds$ which are not satisfied by $S^*$ (such that $u^* \le u$). Observe that $S_1$ defines a solution of cost $c(S_1) + 2^{i-1} \cdot u_1$ to the prize-collecting problem with penalty function $\pi_{2^{i-1}}$. Since $S_1$ was chosen by $\PC$ for that penalty, Property~\ref{asmp:FW_PrizeCollecting} yields 
	
	\begin{equation}
		\label{eq:FW_PrizeCollectingAugmentation}
		c(S_1) + 2^{i-1}\cdot u_1 \le \gamma \cdot (c(S^\ast) +2^{i-1}u^*),  
	\end{equation}
	where the inequality used the fact that $S^\ast$ also defines a solution for the prize-collecting problem.
	
	If $S=S_1$, then since $u_1 \ge \gamma u$, we have 
	\[ 
	    c(S_1)+ \gamma \cdot 2^{i-1} u^* \le c(S_1)+\gamma\cdot 2^{i-1}u \le c(S_1)+ 2^{i-1}u_1. 
	\]
	Combining the previous equation with Eq.~\eqref{eq:FW_PrizeCollectingAugmentation} yields $c(S_1)\le \gamma \cdot c(S^\ast)$, as required.
	
	The second case is that $S=S_2$. In this case, we repeat the argument of Eq.~\ref{eq:FW_PrizeCollectingAugmentation} for $S_2$ with penalty $2^{i}$: 
	\begin{align*}
	 \gamma \cdot \left(c(S^\ast)+2^{i}\cdot u^*\right) 
	 & \ge c(S_2)+2^{i}\cdot u_2 
	 = c(S_2)+2^{i}\gamma u +2^{i}\cdot (u_2 - \gamma u)\\
	& \ge c(S_2)+2^{i}\gamma u + 2^{i}\cdot \frac{u_2-u_1}{2} \quad \left(\text{since } \gamma u < \frac{u_1+u_2}{2} \right).
	\end{align*}
	This yields $c(S_2)\le \gamma \cdot c(S^\ast) + \gamma 2^{i}(u^*-u) + 2^{i} \cdot \frac{u_1-u_2}{2} \le  \gamma \cdot c(S^\ast) + 2^{i} \cdot \frac{u_1-u_2}{2}$. Now, due to Eq.~\ref{eq:FW_PrizeCollectingAugmentation}, we have 
\[ 	\gamma \cdot c(S^\ast) \ge \underbrace{c(S_1)}_{\ge 0} + 2^{i-1}(u_1 - \gamma u^*) \ge 2^{i-1} (u_1 - \gamma u) \ge 2^{i-1}\cdot \frac{u_1 - u_2}{2}. \]
    We thus get that $c(S_2) \le 3\gamma \cdot c(S^\ast)$, completing the proof of the second claim.

\end{proof}

\begin{proof}[Proof of Lemma \ref{lem:FW_BoundingFrameworkCost}]
    {\bf Proof of Item 1. }First, we claim that the total cost of the first $i-1$ iterations
	is at most $(6\gamma + 2)\cdot \offc_{i-1}$.
	
	To prove this claim, consider that at any point during the algorithm, the value of $\onc$ is exactly the total cost incurred by all instances of $\ON$ in Line \ref{line:FW_PathConnect} until that point. Thus, $\onc_{i-1}$ is the cost incurred in the first $i-1$ iterations due to Line \ref{line:FW_PathConnect}.	Since at the beginning (and end) of each iteration it holds that $\onc<2\offc$, we have $\onc_{i-1} \le 2\offc_{i-1}$.
	
	As for the cost of adding items from calls to $\Partial$ in Line \ref{line:FW_AddPartial},
	let $S_{1},...,S_{\ell}$ be the solutions added in Line \ref{line:FW_AddPartial} in the first $i-1$ iterations, in order of their addition. Observe that $c(S_{j})\le 3\gamma \offc_{i-1}\cdot2^{-(\ell-j)}$.
	Thus, the sum of costs of those solutions is at most $6\gamma \offc_{i-1}$.
	
	Now, we claim that the cost of iteration $i$ is at most $O(1) \cdot \OPT+6\gamma \cdot \offc_{i-1}$. To prove this, consider the cost of Line \ref{line:FW_PathConnect} in iteration $i$, in which the online algorithm $\ON$ serves some request $\req_i$. The request $\req_i$ is a part of a phase of requests which are served by this specific instance of $\ON$ (between two subsequent resets of $\ON$ in Line \ref{line:FW_RestartOnline}); call the requests of this phase $\reqs'$. The requests of $\reqs'$ are all served by $\ON$ in the modified input, in which the cost of some set of  elements $S_0$ is set to $0$ (specifically, $S_0$ is the value of the variable $S$ immediately after its augmentation with $S_{\ell}$).
	Using Property~\ref{asmp:FW_SubsetCompetitive}, we have that $\ON({\req_i}) \le O(1)\cdot \OPT'$, where $\ON({\req_i})$ is the cost incurred by $\ON$ when serving $\req_i$, and $\OPT'$ is the optimal solution for $\reqs'$ 
	(with the cost of elements in $S_0$ set to $0$). Clearly, $\OPT' \le \OPT$, which yields that $\ON({\req_i}) \le O(1) \cdot \OPT$.
	
	As for the solution added in Line \ref{line:FW_AddPartial}, its cost is
	at most $3\gamma \offc_{i}=3\gamma \onc_i$. Consider that $\onc_i=\onc_{i-1}+\ON({\req_i})\le2\offc_{i-1}+ O(1)\cdot \OPT$.
	This completes the proof of the first item. 
	
	{\bf Proof of Item 2. }First, we claim that the total cost of $\Partial$ solutions from iteration $i+1$ onward is
	at most $6\gamma \cdot \sum_{j=0}^{m}\ON_{j}$.
	
	Observe that for every $j$ such that $1\le j\le m$, we have that $\offc_{i_j^\star} \ge 2\offc_{i_{j-1}^\star}$. Thus, 
	\[
	    \offc_{i_j^\star} - \offc_{i_{j-1}^\star} \ge \offc_{i_{j-1}^\star} = \offc_{i_j^\star} - \left(\offc_{i_j^\star} - \offc_{i_{j-1}^\star}\right), 
	\] 
	and therefore $ \offc_{i_j^\star} \le 2\left(\offc_{i_j^\star} - \offc_{i_{j-1}^\star}\right)=2\left(\onc_{i_j^\star} - \onc_{i_{j-1}^\star}\right) = 2\cdot \ON_{j-1}$. 
	
	Consider that for every $j$ such that $0\le j \le m$, the cost of the $\Partial$ solution bought in iteration $i_j^\star$ is at most $3\gamma \cdot \offc_{i_j^\star}$. Thus, the total cost of $\Partial$ solutions bought from iteration $i+1$ onward is at most
	\[6\gamma \cdot \sum_{j=0}^{m-1} \ON_{j} \le 6\gamma \cdot \sum_{j=0}^{m} \ON_{j}. \] 
    Which completes the proof of the claim. 
    
    Combining this claim with the definition of $\ON_j$ for every $j$, we have that the total cost of the algorithm from iteration $i+1$ onwards is at most $(6\gamma + 1)\sum_{j=0}^m \ON_j$. Now, we claim that $\sum_{j=0}^{m-2} \ON_j \le \ON_{m-1}$; this would imply that $\sum_{j=0}^{m} \ON_j \le 3\max\{ \ON_{m-1},\ON_m \}$, which completes the proof.
    
    To show that $\sum_{j=0}^{m-2} \ON_j \le \ON_{m-1}$, observe that for every $j\in[m-1]$, it holds that $\ON_{j}=\offc_{i_{j+1}^{\star}}-\offc_{i_{j}^{\star}}$.
	Summing over $j$ from $0$ to $m-2$ yields a telescopic sum:
	\[
	\sum_{j=0}^{m-2}\ON_{j}=\offc_{i_{m-1}^{\star}}-\offc_{i_{0}^{\star}}\le\offc_{i_{m-1}^{\star}}.
	\]
	
	Now, using the same property for $j=m-1$, we get $\ON_{m-1}=\offc_{i_{m}^{\star}}-\offc_{i_{m-1}^{\star}}\ge2\offc_{i_{m-1}^{\star}}-\offc_{i_{m-1}^{\star}}=\offc_{i_{m-1}^{\star}}$. 
	Thus, we get
	\[
	\ON_{m-1}\ge\sum_{j=0}^{m-2}\ON_{j}. 
	\]
	
\end{proof}

\section{Online Priority Steiner Forest with Predictions}
\label{sec:PSF}

A generalization of the Steiner forest problem is priority Steiner forest. 
In this problem, each edge $e \in E$ has an associated integer priority $\pry{e} \in [\npry]$. 
In addition, every request $\req\in \reqs$ contains, in addition to the pair of terminals to connect, a priority demand which we also denote by $\pry{r} \in [\npry]$. 
A feasible solution for this problem is such that for every request $\req$ of terminals $s,t$ there exists a path connecting $s,t$ comprising only edges of priority at least $\pry{r}$. Throughout the discussion of the priority Steiner forest problem, we partition the input $\reqs = \pr{\req_1, \cdots, \req_{\pa{\reqs}}}$ into $\pc{\reqs_j}_{j=1}^{\npry}$, where $\reqs_j$ is the subset of requests of priority $j$.

To define errors in priority Steiner forest, we must define the matching cost between two requests $\req_1$ and $\req_2$. If $\pry{\req_1} \neq \pry{\req_2}$, we say that the matching cost is infinite. In other words, we do not allow matching requests of different priorities. If both $\req_1$ and $\req_2$ have the same priority $j$, then we define the matching cost between $\req_1$ and $\req_2$ to be the distance between the two pairs, as previously defined for Steiner forest. However, this distance is not computed over the original graph $G$, but rather its subgraph $G_j$, which only contains the edges in $G$ which have priority at least $j$. Note the intuition for this definition -- there could be some very low priority class in which all points are very close, but this does not ameliorate the error for higher-priority requests. 

Divide $\reqs$ into $\reqs_1, \cdots, \reqs_{\npry}$ according to priority class, and divide $\preds$ into $\preds_1,\cdots, \preds_{\npry}$ in a similar way. Consider the following algorithm for priority Steiner forest with predictions:
\begin{enumerate}
    \item Run ${\npry}$ parallel instances of the algorithm for Steiner forest with predictions, such that the $j$'th instance runs on the graph $G_j$ (here again $G_j$ is the subgraph of edges of priority at least $j$), and is given the prediction $\preds_j$.
    
    \item Upon the release of a request of priority $j$, send it to the $j$'th instance.
    
    \item At any point during the sequence, maintain the union of the solutions held by all ${\npry}$ instances.
\end{enumerate}

\begin{thm}
    \label{thm:SF_PSF_Competitiveness}
    For every error $\err = (\Delta, D)$, the algorithm above for priority Steiner forest with predictions yields the following:
    \[ 
        \ALG \le O\pr{\npry\log \pr{\frac{\min\pc{\pa{\reqs},\Delta}}{\npry}+2}}\cdot \OPT + O(1)\cdot D
    \]
\end{thm}
\begin{proof}
    Observe that for every priority class $j$, restricting the error $\err$ to $\reqs_j$ and $\preds_j$ induces a Steiner forest error $\err_j$ for the $j$'th instance, such that $\err_j = (\Delta_j, D_j)$ where $\sum_{j=1}^{\npry} \Delta_j = \Delta$ and $\sum_{j=1}^{\npry} D_j = D$. Now, one can apply Theorem \ref{thm:RSF_Competitiveness} for every $j$, and obtain $\ALG_j \le O\pr{\log \pr{\min\pc{\pa{\reqs_j},\Delta_j}+2}}\cdot \OPT_j + O(1)\cdot D_j$, where $\ALG_j $ is the cost of the Steiner forest with predictions algorithm on the requests of $\reqs_j$ (given prediction $\preds_j$), and $\OPT_j$ is the cost of the optimal solution for that same Steiner forest instance.
    
    Summing over $j$, we obtain 
    \begin{align*}
        \ALG 
        &\le \sum_{j=1}^{\npry} \ALG_j \\
        &\le \sum_{j=1}^{\npry} O\pr{\log \pr{\min\pc{\pa{\reqs_j},\Delta_j}+2}}\cdot \OPT_j + \sum_{j=1}^{\npry} O(1)\cdot D_j \\
        &\le \sum_{j=1}^{\npry} O\pr{\log \pr{\min\pc{\pa{\reqs_j},\Delta_j}+2}}\cdot \OPT +  O(1)\cdot D \\
        &\le O\pr{\npry \log \pr{\frac{\sum_{j=1}^{\npry} \pr{\min\pc{\pa{\reqs_j},\Delta_j}+2}}{\npry} } } \cdot \OPT + O(1)\cdot D \\
        &\le O\pr{\npry \log \pr{\frac{ \min\pc{\sum_{j=1}^{\npry} \pa{\reqs_j}, \sum_{j=1}^{\npry} \Delta_j}}{\npry} +2 } } \cdot \OPT + O(1)\cdot D \\
    \end{align*}
    \begin{align*}
        &= O\pr{\npry \log \pr{\frac{ \min\pc{\pa{\reqs}, \Delta}}{\npry} +2 } } \cdot \OPT + O(1)\cdot D \\
    \end{align*}
    
    where the third inequality is due to the fact that $\OPT_j \le \OPT$ for every $j$ and that $\sum_{j=1}^b D_j = D$, and the fourth inequality is due to Jensen's inequality and the concavity of $\log$. This completes the proof of the theorem.
\end{proof}

\section{Online Facility Location Algorithm~\texorpdfstring{\cite{Fotakis07}}{} is {\sc not} Subset Competitive}
\label{sec:SubsetCompetitivenessExample}

In this section, we show that the algorithm of Fotakis for online facility location \cite{Fotakis07} is not subset-competitive. This algorithm is denoted by $\OFL$, and described in Section \ref{sec:FL}. This deterministic algorithm is $O(\log |R|)$-competitive. To show that it is not subset competitive, we show that for every constant $c$ there exists an input $R$ and a subsequence $R'\subseteq R$ such that $\ON(R')>c\cdot \log |R'|$.

We consider the lower bound of $\Omega(\frac{\log |R|}{\log \log |R|})$ of \cite{Fotakis08}. The lower bound consists of an input on which every algorithm is $\Omega(\frac{\log k}{\log \log k})$-competitive, where $k$ is such that $|R| \le k$. 

The input is composed of a full binary tree of height $m=\frac{\log k}{\log \log k}$, where the cost to open a facility at each node is $f$. The weights on the edges of the tree are equal for each level, and decrease exponentially by a factor of $m$ in each level, where the edges going out of the root have a weight of $\frac{f}{m}$. The request sequence is composed of $m+1$ phases, where the $i$'th phase releases $m^{i-1}$ requests on a node $v_i$ at depth $i-1$. The first phase releases a single request on the root $v_0$; For each subsequent phase $i$, the node $v_i$ is the child node of $v_{i-1}$ such that the algorithm did not open a facility in the subtree of that child node (if the algorithm opened a facility in the subtrees of both child nodes, $v_{i}$ is an arbitrary child node of $v_{i-1}$).

Now, consider the operation of the online algorithm of \cite{Fotakis07} on this input. Upon the single request of the  first phase, the algorithm opens a facility at the root $v_1$, and connects that request. The potential of all points at the end of this phase is zero. 

At each subsequent phase $i$, the algorithm connects all requests except for the last request to the facility at $v_{i-1}$. At the last request of the phase, the potential of $v_i$ at this point is $m^{i-1} \cdot \frac{f}{m^{i-1}} = f$, and thus a facility is opened at $v_i$. Note that the total connection cost of this phase is $(m^{i-1}-1) \cdot \frac{f}{m^{i-1}}$, so at most $f$.

Overall, the algorithm opened facilities at $v_1,\cdots, v_{m+1}$, at an opening cost which is more than its connection cost, and is thus at least half its total cost. From the guarantee of the lower bound, the opening cost of the algorithm is thus $\Omega({\frac{\log k}{\log \log k}})\cdot \OPT$, where $k$ is at least the number of requests, and thus $k\ge m^m$.

Now, consider the subsequence $R'$ which consists of the last request in each phase. There are $m+1$ such requests, but the entire facility opening cost of the algorithm is incurred on those requests. Thus, $\OFL(R') = \Omega(\frac{m\log m}{\log m})=\Omega(m)=\Omega(|R'|)$. In particular, $\OFL(R')=\omega(\log |R'|)$, proving that the algorithm is not subset-competitive.

\section{Lower Bounds for Online Algorithms with Predictions}
\label{sec:lowerbound}

In this section, we show additional lower bounds for the Steiner tree, Steiner forest and facility location in the online algorithms with predictions setting.

\subsection{Steiner Tree}

We start with a lower bound for Steiner tree, which also applies to Steiner forest.

Suppose we are given the parameters $n,k,\Delta_1,\Delta_2$. We define a $(n,k,\Delta_1,\Delta_2)$-adversary to be an adversary for Steiner tree with predictions that gives a prediction $\preds$ and a request sequence $\reqs$ such that $|\preds|=n$, $|\reqs|=k$, $|\preds \backslash \reqs|= \Delta_1$, $|\reqs\backslash \preds| = \Delta_2$. The following theorem yields the desired lower bound.

\begin{thm}
	\label{thm:LB_SteinerTree}
	For every $n,k,\Delta_1,\Delta_2$, there exists a $(n,k,\Delta_1,\Delta_2)$-adversary for Steiner tree with predictions such that every randomized algorithm on this instance is $\Omega(\min (k,\Delta))$-competitive, where $\Delta = \Delta_1 + \Delta_2$.
\end{thm}

For this theorem, we use the lower bound for the online Steiner tree problem, due to Imase and Waxman \cite{ImaseW91}, which we restate for the rooted version of Steiner tree.

\begin{thm}
	\label{thm:LB_ST_ImaseWaxmanLB}
	For every $m$, there exists a graph $G_m = (V,E)$ of $m^2 + m$ nodes, a root $\rho \in V$ and an oblivious adversary $\ADV$ that releases a sequence of requests $\reqs$ such that $|\reqs|=m$ on which every algorithm is $\Omega(\log m)$ competitive.
\end{thm}
\begin{proof}
    The construction of \cite{ImaseW91} consists of an adversary $I_i$ for every $i$ such that every algorithm is $\Omega(i)$ competitive for $I_i$. The instances are built recursively -- $I_0$ consists of a single edge between the root and another node. Each $I_i$ is constructed from $I_{i-1}$ by breaking each edge into a diamond -- that is, adding two nodes and replacing the edge with $4$ edges. 
    
    Therefore, the number of edges in $I_i$ is $4^i$. The number of nodes in $I_i$ is the number of nodes in $I_{i-1}$ plus twice the number of edges in $I_{i-1}$ -- thus, at most $4^i$ for every $i \ge 1$. The number of requests released in $I_i$ is exactly $2^i$. Thus, the number of nodes is at most the square of the number of requests.
    
    Now, to obtain the desired instance, set $i = \floor{\log m}  $. The number of nodes will be $4^i \le m^2$. Pad the number of nodes to exactly $m^2 +m$ by adding copies of the root $\rho$. Pad the number of requests to exactly $m$ by requesting some of those copies of $\rho$.
\end{proof}

Fix any values of $n,k,\Delta_1,\Delta_2$. Define $\ell = \frac{n+k-\Delta_1 - \Delta_2}{2}$, the size of the intersection of $\reqs$ and $\preds$ in the adversary. Lemmas  \ref{lem:LB_Delta2Adversary} and \ref{lem:LB_IntersectionAndDelta1Adversary} imply Theorem \ref{thm:LB_SteinerTree}.

\begin{lem}
    \label{lem:LB_Delta2Adversary}
    There exists an oblivious $(n,k,\Delta_1,\Delta_2)$-adversary on which every algorithm is $\Omega(\log \Delta_2)$-competitive.
\end{lem}

\begin{lem}
    \label{lem:LB_IntersectionAndDelta1Adversary}
    There exists an oblivious $(n,k,\Delta_1,\Delta_2)$-adversary on which every algorithm is $\Omega(\log (\min (\ell, \Delta_1)))$-competitive.
\end{lem}

\begin{proof}[Proof of Theorem \ref{thm:LB_SteinerTree}]
    If $\Delta_2 \ge \min(\ell,\Delta_1)$, then $2\Delta_2 \ge \min(\ell +\Delta_2,\Delta_1+\Delta_2) = \min (k,\Delta)$. Thus, the adversary of Lemma \ref{lem:LB_Delta2Adversary} yields the desired result.
    
    Otherwise, $\Delta_2 \le \min(\ell, \Delta_1)$. Using $\Delta_2 \le \ell$, we have that $\ell \ge \frac{\ell + \Delta_2}{2}=\frac{k}{2}$. Using $\Delta_2 \le \Delta_1$, we have that $\Delta_1 \ge \frac{\Delta_1+\Delta_2}{2}=\frac{\Delta}{2}$. Thus, $\min(\ell,\Delta_1) \ge \frac{\min(k,\Delta)}{2} $, which completes the proof of the theorem.
\end{proof}

It remains to prove Lemmas \ref{lem:LB_Delta2Adversary} and \ref{lem:LB_IntersectionAndDelta1Adversary}.

\begin{proof}[Proof of Lemma \ref{lem:LB_Delta2Adversary}]
    We construct the adversary in the following way. The graph on which the requests are released is the graph $G_m$ of Theorem \ref{thm:LB_ST_ImaseWaxmanLB} for $m=\Delta_2$, with multiple copies of the root $\rho$. The $n$ predictions of $\preds$ all arrive on copies of $\rho$. The input sequence $\reqs$ starts with $\ell$ requests on predicted copies of $\rho$ (which have a connection cost of $0$ for the optimal solution). We now use the Imase-Waxman adversary on $G_m$ to release the remaining $k-\ell = \Delta_2$ requests, on which any algorithm is $\Omega(\log \Delta_2)$ competitive.
\end{proof}

\begin{proof}[Proof of Lemma \ref{lem:LB_IntersectionAndDelta1Adversary}]
    We construct the adversary in the following way. The graph on which the requests are released is the graph $G_m$ for $m = \min (\ell, \sqrt{n-\ell})$, with multiple copies of the root node $\rho$. The prediction $\preds$ consists of all $m^2+m$ nodes of $G_m$ (note that $m^2+m \le n$), as well as $n-(m^2+m)$ copies of $\rho$.
    
    The adversary starts the request sequence by requesting $\Delta_1$ unpredicted copies of $\rho$. Serving these requests costs $0$ to the optimal solution. It then requests $\ell - m$ \emph{predicted} copies of $\rho$ (note that there exist $\ell -m$ predicted copies of $\rho$, since $n-(m^2+m)\le \ell -m$.
    
    It remains to release $m$ predicted requests. The adversary now calls the Imase-Waxman adversary of Theorem \ref{thm:LB_SteinerTree} to release these $m$ requests on $G_m$. 
    
    Every algorithm is $\Omega(\log(m))$ competitive on this adversary. But $m = \min(\ell,\sqrt{n-\ell}) \ge \sqrt{\min(\ell,n-\ell)}=\sqrt{\min(\ell,\Delta_1)}$, and thus $\log m \ge \frac{\log (\min(\ell,\Delta_1))}{2}$. This completes the proof.
\end{proof}

\subsection{Facility Location}
We now give a similar result to the previous lower bound for facility location. 

\begin{thm}
    \label{thm:LB_FacilityLocation}
    For every $n,k,\Delta_1,\Delta_2$, there exists a $(n,k,\Delta_1,\Delta_2)$-adversary for facility location with predictions such that every randomized algorithm on this instance is $\Omega\left(\frac{\log \min (k,\Delta)}{\log\log \min (k,\Delta)}\right)$-competitive, where $\Delta = \Delta_1 + \Delta_2$.
\end{thm}

The following lower bound for online facility location is due to Fotakis~\cite{Fotakis08}. 

\begin{thm}
    \label{thm:LB_FL_FotakisLB}
	For every $m$, there exists a graph $G_m = (V,E)$ of $m^2 + m$ nodes and an oblivious adversary $\ADV$ that releases a sequence of requests $\reqs$ such that $|\reqs|=m$ on which every algorithm is $\Omega\left(\frac{\log m}{\log\log m} \right)$ competitive.
\end{thm}
\begin{proof}
    The adversary given in \cite{Fotakis08} gives a graph which is a complete binary tree of some depth $h$. For $a=\frac{\log m}{\log \log m}$, each node of depth $b$ in the tree has $a^b$ copies. The adversary first requests the root, then requests the $a$ copies of a node of depth 1, then the $a^2$ copies of a node of depth $2$, and so on. The total number of requests released is at most $m$. The number of nodes at each depth $b$ is $2^b \cdot a^b \le (a^b)^2$, and thus the number of nodes in the graph is at most $m^2$. The analysis of \cite{Fotakis08} guarantees that every algorithm is $\Omega\left(\frac{\log m}{\log\log m} \right)$ competitive on this instance.
    
    To reach the exact desired number of nodes and requests, we pad the number of nodes to exactly  $m^2 + m$ by adding nodes of facility opening cost $0$ at infinite distance from the rest of the construction. We pad the number of requests to $m$ by requesting these nodes added in the padding.
\end{proof}

Fix any values of $n,k,\Delta_1,\Delta_2$. Define $\ell = \frac{n+k-\Delta_1 - \Delta_2}{2}$, the size of the intersection of $\reqs$ and $\preds$ in the adversary. The following two lemmas are analogous to the lemmas from the Steiner tree lower bound, and imply Theorem \ref{thm:LB_FacilityLocation} in the same way that their analogues implied Theorem \ref{thm:LB_SteinerTree}.

\begin{lem}[Analogue of Lemma \ref{lem:LB_Delta2Adversary}]
    \label{lem:LB_FL_Delta2Adversary}
    There exists an oblivious $(n,k,\Delta_1,\Delta_2)$-adversary on which every algorithm is $\Omega\left(\frac{\log \Delta_2}{\log\log \Delta_2} \right)$-competitive.
\end{lem}

\begin{proof}
    The proof is similar to the proof of Lemma \ref{lem:LB_Delta2Adversary}. The instance comprises the graph $G_m$ as defined in Theorem \ref{thm:LB_FL_FotakisLB}, for $m=\Delta_2$, and from multiple nodes of facility cost $0$ and distance $\infty$ from every other node. The $n$ predictions are all of nodes at distance $\infty$, and so are the first $\ell$ requests. The following $\Delta_2$ requests are given to the adversary of $G_m$.
\end{proof}

\begin{lem}[Analogue of Lemma \ref{lem:LB_IntersectionAndDelta1Adversary}]
    \label{lem:LB_FL_IntersectionAndDelta1Adversary}
    There exists an oblivious $(n,k,\Delta_1,\Delta_2)$-adversary on which every algorithm is $\Omega\left(\frac{\log \min (\ell, \Delta_1)}{\log \log \min (\ell, \Delta_1)}\right)$-competitive.
\end{lem}

\begin{proof}[Proof of Lemma \ref{lem:LB_FL_IntersectionAndDelta1Adversary}]
    The proof is similar to the proof of Lemma \ref{lem:LB_IntersectionAndDelta1Adversary}, with the only difference being as in the proof of Lemma \ref{lem:LB_Delta2Adversary} -- we use nodes at distance infinity with facility cost $0$ to pad predictions and requests, in lieu of the copies of the root used for the Steiner tree problem.
\end{proof}

\section{Lower Bound for Online Metric Matching with Predictions}
\label{sec:MMLB}
In this section, we consider the online metric matching problem, and show that in contrast to the network design problems in this paper, it has no tolerance for errors. 
Specifically, we show that while the offline problem can be solved optimally, the online problem with predictions performs as badly as the online problem without predictions, even in the presence of \emph{only 2 prediction errors}.

In the online metric matching problem, one is given a metric space $M = (X,d)$, as well as a set of $k$ ``blue'' points in $M$. 
One after the other, ``red'' points are released, and the algorithm must match each arriving red point to an unmatched blue point, at a cost which is the distance in $M$ between the two points. 
In this online version, after the algorithm matches two points, that matching cannot be undone.

The best known lower bounds for this problem, both in the randomized and deterministic settings, are both on a metric induced by an unweighted star graph with $k+1$ spokes (each edge has a cost of $1$), where the blue points are on $k$ spokes of the graph.
The first red point is released on the remaining spoke, which is matched to a blue point by the algorithm.

For a deterministic algorithm, the lower bound proceeds as follows: the adversary releases a red point on the blue point to which the algorithm has matched the previous red point, and repeats this for $k-1$ iterations. 
In the end of the instance, there remains only a single blue point $v$ on which no red point has been released.
The cost of the algorithm for this sequence is $2k$, as the cost of each matching is $2$.
The optimal solution, on the other hand, has a cost of $2$ - it matches the first red request to $v$ at a cost of $2$, and matches each remaining red point to its colocated blue point.
This example thus shows a lower bound of $\Omega(k)$ on the competitive ratio of a deterministic algorithm in this setting.

Now, consider the online algorithm augmented with a prediction for the red points, where the prediction predicts a single red point on every blue point. Clearly, this prediction does not provide any information to the algorithm, and does not improve its performance on the input. However, it has $(2,0)$ error - one unpredicted red point arrives, and one predicted red point does not arrive (as the remaining points arrive exactly, the matching cost is $0$). This example is shown in Figure \ref{fig:MMLB_Figure}.

\begin{figure}
    \begin{center}
    \includegraphics{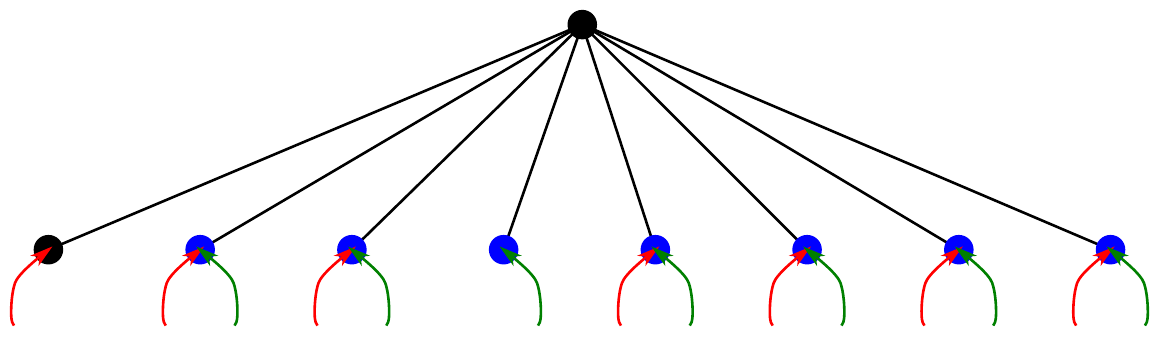}
    \end{center}
    \caption{\footnotesize This figure visualizes the lower bound for online metric matching, which applies also for metric matching with predictions. The figure shows an unweighted star graph, where all but one of its spokes have blue points. The green arrows are the predicted red points, while the red arrows are the actual red points in the input.}
    \label{fig:MMLB_Figure}
\end{figure}

In the randomized setting, the outcome is similar. In the randomized setting, the adversary is oblivious, releasing a red point on a random blue point, chosen uniformly from the blue points on which no red point has been released. A standard survivor-game analysis yields $\Omega(\log k)$ competitiveness on this instance; as before, predicting all the blue points as input does not help competitiveness, but has an error of $(2,0)$.

\end{document}